\documentclass[11pt]{article}

\usepackage{amsfonts}
\usepackage{latexsym,amssymb,amsmath,amsthm}
\usepackage{thmtools} 
\usepackage{thm-restate}
\usepackage{graphicx}
\usepackage[margin=1.0in]{geometry}
\usepackage[sort]{natbib}
\usepackage{tabularx}
\usepackage{caption}

\usepackage{tikz}
\usepackage{pgfplots}
\usepackage{enumitem}
\allowdisplaybreaks

\pgfplotsset{compat=1.13}

\newtheorem{theorem}{Theorem}
\newtheorem{obs}[theorem]{Observation}
\newtheorem{observation}[theorem]{Observation}
\newtheorem{claim}[theorem]{Claim}

\newtheorem{lemma}[theorem]{Lemma}

\newtheorem{definition}[theorem]{Definition}

\newtheorem{example}[theorem]{Example}

\newcommand{\eps}{\epsilon}

\newcommand{\I}{\mathbb{I}}

\newcommand{\ignore}[1]{}

\newcommand{\win}{\textsc{success}}
\newcommand{\winn}{\textsc{win}}

\newcommand{\opt}{\textsc{opt}}
\newcommand{\live}{\textsc{here}}
\newcommand{\acc}{\textsc{acc}}
\newcommand{\rej}{\textsc{rej}}

\newcommand{\dead}{\textsc{gone}}
\newcommand{\first}{\textsc{first}}
\newcommand{\mir}{\textsc{rob}}

\newcommand{\hist}{\textsc{hist}}
\newcommand{\fut}{\textsc{fut}}

\newcommand{\numsofar}{\numbersofar}

\newcommand{\numbersofar}{K}
\newcommand{\sofarval}{k}

\newcommand{\arrtimevecvals}{\vec{a}}
\newcommand{\deptimevecvals}{\vec{d}}
\newcommand{\rankvecvals}{\vec{r}}
\newcommand{\arrtimeval}{a}
\newcommand{\deptimeval}{d}
\newcommand{\rankval}{r}
\newcommand{\arrtime}{A}
\newcommand{\deptime}{D}
\newcommand{\rank}{R}
\newcommand{\staytime}{L}
\newcommand{\staytimeval}{\ell}
\newcommand{\arrtimedist}{\mathcal{A}}

\newcommand{\staydist}{\mathcal{L}}

\newcommand{\pol}{\textsc{pol}}

\newcommand{\thresh}{\theta}
\newcommand{\pois}{\delta}
\newcommand{\sdsp}{stochastic departure secretary problem}

\newcommand{\growingmid}{\mathrel{}\middle|\mathrel{}}

\begin{document}

\title{How to Hire Secretaries with Stochastic Departures}
\author{Thomas Kesselheim\thanks{Institute of Computer Science, University of Bonn, 53115 Bonn, Germany. Email: \texttt{thomas.kesselheim@uni-bonn.de}} \and Alexandros Psomas\thanks{Simons Institute for the Theory of Computing, Berkeley CA, USA. Email: \texttt{alexpsomi@cs.berkeley.edu}} \and  Shai Vardi\thanks{Krannert School of Management, Purdue University, West Lafayette, IN, USA. Email: {\tt  svardi@purdue.edu}. }}
\date{}

\maketitle

\begin{abstract}
	We study a generalization of the secretary problem, where  decisions do not have to be made immediately upon candidates' arrivals.
	After arriving, each candidate stays in the system  for some (random) amount of time and then leaves, whereupon the algorithm has to decide irrevocably whether to select this candidate or not. The goal is to maximize the probability of selecting the best candidate overall. We assume that the arrival and waiting times are drawn from known distributions.

	Our first main result is a characterization of the optimal policy for this setting. We show that when deciding whether to select a candidate it suffices to know only the time and the number of candidates that have arrived so far. Furthermore, the policy is monotone non-decreasing in the number of candidates seen so far, and, under certain natural conditions, monotone non-increasing in the time.
	Our second main result is proving that when the number of candidates is large, a single threshold policy is almost optimal.
\end{abstract}

\section{Introduction}

The secretary problem is an online selection problem whose origin is still being debated, but is usually attributed to \citet{Gardner60},~\citet{Lindley61} or~\citet{Dynkin63}. In the original secretary problem $n$ candidates are interviewed in uniformly random order. Only one candidate can be hired, and the goal is to maximize the probability of hiring the best candidate. After each interview, the interviewer must make an immediate and irrevocable decision whether to hire the candidate, and is only allowed to make pairwise comparisons between candidates that have already been interviewed. The optimal solution is to wait for some threshold, whose exact value (for specific $n$) is determined by backward induction (e.g.,~\cite{Lindley61,GM66}),  and the probability of hiring the best candidate is asymptotically $1/e$.

In the past few years the secretary problem and variations thereof  have received a lot of attention because of their applications to, among other things, online auctions, e.g.,~\citet{BIK07,BIKK08,EHKS18}. 
An attractive common occurrence in secretary problems (and stopping problems in general) is that the optimal solution is typically a thresholding algorithm. 
Thresholding algorithms naturally translate to 
posted price mechanisms, which are inherently truthful;   maximizing the value of the chosen candidate(s) is equivalent to  maximizing the social welfare. 
Most of the work on (variations of the) secretary  problem does not depart from the  truly online nature of the original: an irrevocable decision has to be made \emph{immediately} upon seeing the agent. In many situations of interest, however, it is reasonable to assume that the decision does not have to be immediate.  In the classical secretary-hiring scenario, it is reasonable to expect that the candidate will still be available for some time after the interview;  
instead of showing a user an advertisement immediately upon arrival to a website, it could be more profitable to wait before presenting an advertisement (if she is expected to stay on the website for a while).
In this paper we introduce a general model for the secretary problem in which candidates do not immediately leave the system.

\subsection{Main results}
There are $n$ candidates with some total preference order. Each candidate $i$ arrives at a time $\arrtime_i$, drawn i.i.d.\ from some \emph{arrival distribution} $\arrtimedist$. Candidate $i$ stays in the system for some time $\staytime_i \geq 0$, drawn i.i.d.\ from some \emph{waiting distribution} $\staydist$. At time $\deptime_i = \arrtime_i + \staytime_i$, the algorithm is informed that candidate $i$ is about to leave. The algorithm then has to make the irrevocable decision whether to accept or reject candidate $i$. The goal is to maximize the probability of accepting the best candidate. We call this problem the \emph{\sdsp}. The standard secretary problem is recovered if $\staytime_i = 0$ with probability $1$ for all $i$, and any   $\arrtimedist$ that has no point mass. For other distributions $\staydist$, the algorithm can possibly take advantage of having seen further candidates between time $\arrtime_i$ and $\deptime_i$ that serve as a point of comparison. 
We note that our model subsumes sliding-window models from the literature (e.g.,~\cite{HK15}, see below for details), in which $\staytime_i$ is a fixed, constant value.\footnote{Technically, in these models, the $i^{\text{th}}$ candidate arrives at time $\frac{i}{n}$; the arrival times are not drawn from some distribution. It is well known (e.g.,~\cite{Vardi15}), that for large $n$, their model is asymptotically equivalent to uniform arrivals.} 
We remark that all of our results hold in the model in which the algorithm is allowed to make a decision at any time before the candidate leaves, as the algorithm can never increase its success probability by accepting earlier than it needs to. Furthermore, the algorithm only accepts a candidate if he is the best candidate observed so far.

Our first main result is an optimal selection rule for an arbitrary  arrival distribution  and an arbitrary departure distribution.

\begin{theorem}(Informal)\label{thm:opt}
	There exists an optimal policy for the \sdsp\, that decides whether to accept or reject at time $t$ whenever a best-so-far candidate leaves, such that
	\begin{enumerate}[leftmargin=*]
		\item[(1)] the decisions depend only on $t$ and $\numsofar_t$, the number of candidates that have arrived until time $t$,
		\item[(2)] the decision is monotone non-decreasing in $\numsofar_t$: fixing $t$, there exists some $k$ such that if $\numsofar_t\leq k$, it rejects, and if $\numsofar_t > k$, it accepts.
		\item[(3)] if the arrival distribution is uniform, the decision is monotone non-increasing in $t$:  fixing $\numsofar_t$, there exists some $\theta$ such that if $t\leq\theta$, it accepts, and if $t > \theta$, it rejects. There are distributions for which this monotonicity does not hold.
	\end{enumerate}
\end{theorem}

In the classical secretary problem ($\staytime_i = 0$), the optimal rule only depends on the number of candidates seen thus far but not on the time that has passed. We give examples where stochastic departures make the dependence on both the time and number of candidates  unavoidable (Examples~\ref{ex:needd} and~\ref{ex:unif} in Section~\ref{sec:optimal stopping}).

Our second main result is showing that when the arrival distribution is continuous and  the number of candidates is large, the decisions do not have to depend on the number of candidates anymore. So, a \emph{single-threshold policy} achieves a good success probability.
\begin{theorem}
	\label{theorem:timethresholding}There exists a threshold policy $\theta$ for the \sdsp{} with a continuous arrival distribution that is independent of $n$ and $K_t$, that accepts a candidate on departure if and only if he is the best so far and $t>\theta$. In the limit for large $n$, its (asymptotic) success probability matches the one of an optimal policy. 
\end{theorem}

When $n$ is small, even if the departure time is immediate, policies that set a time threshold are far from optimal (Example~\ref{ex3} in Appendix~\ref{app:examples}). Theorem~\ref{theorem:timethresholding} also carries over to the setting in which the arrivals are generated in a Poisson point process; that is, the overall number of candidates $n$ is not fixed but random, see Section~\ref{subsec:poisson arrivals}.

\subsection{Techniques}
When facing a decision, the optimal policy picks a candidate if and only if this gives a higher success probability than rejecting it. The success probability after a rejection again depends on the policy. As a result, when characterizing the optimal policy throughout the paper,  the success probabilities are  frequently bound in different conditional probability spaces.

A pivotal technique to bound these success probabilities is  simulation arguments. A simple such argument is used to prove Claim (1) of Theorem~\ref{thm:opt}, where we show that it is sufficient for the optimal policy to know the time $t$ and number of candidates seen so far $\numsofar_t$. If this is not the case, there are two histories for the same time $t$ and number of candidates $\numsofar_t$ for which the optimal policy makes different decisions. Since the success probabilities when accepting are identical, the success probabilities when rejecting have to be different. We then define a new policy that follows the decisions of the optimal policy for the ``better'' history whenever it sees the ``worse'' one. We show that all the relevant events have identical marginal probabilities, and therefore our new policy does better than the optimal policy on the ``worse'' history, a contradiction. 

Later arguments require a more detailed look at the probability spaces generating the events. For example, to show Claim (2) of Theorem~\ref{thm:opt}, we have to argue that after having seen $k$ candidates if the success probability when accepting the current candidate is higher than the one when rejecting, then the same holds when having seen $k+1$ candidates. In this case, the simulating policy pretends that $k+1$ candidates arrived rather than $k$. It does so by deleting a random future arrival and replacing it by a rejected candidate earlier in the sequence. The complication is that future observations when having seen $k$ candidates are still not identically distributed to the future observations when having seen $k+1$ candidates. We overcome this by coupling with a suitably chosen conditional probability space. As a result, the probability of a future observation can only be smaller by a multiplicative factor, namely the probability of the conditioned upon event. That is, even though this coupling does not give a one-to-one correspondence between the two probability spaces, we can still get an upper bound on the probability of a future observation, which suffices for our purpose.

At first glance, Theorem~\ref{theorem:timethresholding} seems quite straightforward. If $n$ is large, the number of arrivals by time $t$, $\numsofar_t$, is concentrated around its expectation. Therefore, a policy can, with some error, replace $\numsofar_t$ by $\mathbb{E} [ \numsofar_t ]$. While this observation may be correct, it is by far not enough to prove the theorem. For example, it is not clear whether the policy designed this way is a threshold policy for every fixed $n$. Furthermore, the theorem states that a \emph{single} policy is near optimal for all large $n$ \emph{simultaneously}. Therefore, one would also have to prove that the policies for different $n$ have a common limit of some sort.

Our approach is to instead construct an explicit threshold policy. At any time $t$, the policy compares two probabilities: (a) the probability that the overall best candidate arrives by $t$, and (b) the success probability of an optimal policy that only accepts candidates arriving after $t$. Whenever (a) is greater than (b), accept a departing candidate if he is the best so far. Note that by definition (a) is non-decreasing and (b) is non-increasing in $t$, so this policy is a threshold policy.
To prove asymptotic optimality, we observe that the optimal policy actually compares the same probabilities but in the probability space conditioned on $K_t$ and the candidate that is best so far leaving at time $t$. We show that, because of concentration bounds, the conditional probabilities are in most cases close to the unconditional ones. For this reason, the newly constructed policy makes different decisions only when the two options, accept or reject, yield similar success probabilities. 
It then remains to bound the loss in success probability by these errors.
We compare conditional and unconditional probabilities using  simulation arguments. Again, success probabilities conditioning on different numbers of arrivals are compared. The key difference is that  the question now is not whether one of them is bigger than the other, but by how much they differ.

\subsection{Case study: Poisson arrivals and exponential departures}

The most interesting choice arrival and departure distributions  are arguably \emph{Poisson arrivals} and \emph{exponential departures}. The Poisson process has a long history of modeling stochastic arrival processes, e.g.,~\cite{Bajari,Pinker}. Similarly, the exponential distribution is arguably the most popular way to model ``impatience'' or waiting times, e.g.,~\cite{exp1,exp2}. 
Concretely, we would like to compute the optimal threshold for the following setting. Candidates arrive according to a Poisson process with parameter $\delta$; the waiting distribution is exponential with rate $\lambda$. As we mentioned earlier, threshold policies are approximately optimal in this scenario. We give a closed form expression for the probability of success as a function of a threshold. 

Unfortunately, this scenario does not fall within our framework, as the number of arrivals is unknown a-priori. However, we show that for sufficiently large $\delta$, we can leverage our previous insights, in order to (1) show that a threshold policy is approximately optimal, and (2) find a closed form expression for the probability of success as a function of this threshold. See Section~\ref{SEC:EXP} for more details.

\subsection{Related work}

The secretary problem and its variants have received much attention in the later part of the 20th century. 
We refer the reader to~\cite{Freeman83} and~\cite{Ferguson89} for surveys on the classical literature on secretary problems and variations thereof.  More recently, there has been a surge of interest in variations of the secretary problem in the theoretical computer science community, (e.g.,~\citet{HKP04, KRTV13,BIKK08,rubinstein2016beyond})  in large part driven by its applications to online mechanism design, in particular ad-auctions, where users arrive online and are matched to advertisers. Notable variants include the matroid secretary problem (e.g.,~\citet{BIK07,kleinberg2012matroid,Lachish14}) and the prophet secretary (e.g.,~\citet{esa,EHKS18,azar2018prophet}), as well as applications in sequential posted pricing (e.g.,~\citet{CHMS10,bey})  and online trading (e.g.,~\citet{KL18}).

Special cases of not making an immediate decision  have been addressed in the literature:~\citet{Goldys} showed that  the expected rank of the accepted candidate tends to $\approx 2.57$ as $n$ tends to infinity, when one is allowed to choose either the current candidate or the previous one.
This is in contrast to the expected rank of $\approx 3.87$ when one is only allowed to choose the current candidate, shown by~\citet{Chow1964}. In this setting, the value of a candidate is $n-i$ when the $i^{\text{th}}$ best candidate is accepted, as opposed to $1$ if the best candidate is accepted and $0$ otherwise in the classical setting.
The scenario considered by Goldys is sometimes called a ``sliding window''. In the online setting, a sliding window of size $x$ means that after seeing an item (candidate), the algorithm does not need to make a decision until it has seen $x$ more items. Sliding windows have been considered for many online and streaming problems (e.g.,~\cite{JKZ04,BK09,Datar02}). Goldys' setting is a  sliding window of size $1$.

\citet{HK15} considered more general sizes of sliding windows. They give an optimal threshold-based rule for maximizing the probability of accepting the best candidate; give a recursive (non-explicit) formula for the probability of hiring the best candidate using sliding windows of size $n/i$ for constant $i$; and give an asymptotic bound when the window size is at least $n/2$. 
\citet{Vardi15} considered the scenario where candidates arrive $k$ times each, and the arrival order is uniform over the $(kn)!$ possible arrival orders. He gave an optimal threshold-based strategy for accepting the best candidate and computed the success probability for $k=2$, as well as giving upper bounds for the matroid secretary version of this problem. ~\citet{lisa} extended these results to other packing domains.
\citet{Petruccelli81} considered the case when the interviewer is able to recall some candidate from the past with some probability $p>0$. 

We note that in all of the above cases---in contrast to our setting---the optimal policy  depends only on the number of observed candidates, and not on the time at which the decision is made (in fact, in the works above, the two are typically interchangeable).

\section{Model}

A set $\mathcal{S}$ of $n$ candidates arrive and depart over time. For concreteness, we assume that events only occur in the interval $[0,1]$. There is a  total order on  $\mathcal{S}$,  and the goal is to select the best element in this order. For each $i \in \mathcal{S}$, an \emph{arrival time} $\arrtime_i$ is drawn independently from the arrival distribution $\arrtimedist$  and a \emph{waiting time} $\staytime_i$ (denoting how long the candidate  stays in the system) is a non-negative real number drawn i.i.d. from some distribution $\staydist$, which we call the \emph{waiting time distribution}. The results of Section~\ref{sec:optimal stopping} all hold for arbitrary arrival distributions except for Lemma~\ref{lem:t}, which concerns the uniform distribution. In Section~\ref{SECTION:LIMIT}, we assume that the arrival distribution is continuous; i.e.,  has no point masses.
The \emph{departure time} $\deptime_i$ is  $\min\{\arrtime_i+\staytime_i,1\}$.

We index the set $\mathcal{S}$ by arrival time, i.e., by $1, \ldots, n$ in such a way that $\arrtime_1 \leq \arrtime_2 \leq \ldots \leq \arrtime_n$;  the total order can be expressed as a permutation $\pi\colon [n] \to [n]$. The algorithm is successful if it chooses $i \in [n]$ such that $\pi(i) = 1$. The goal is to design an  algorithm that maximizes the probability of selecting the best candidate. 
It will be useful to represent the permutation $\pi$ as a sequence of relative ranks $\rank_1, \rank_2, \ldots, \rank_n$, where $\rank_i \in [i]$ indicates the number of candidates among $1, \ldots, i$ that are at least as good as $i$. Formally, $\rank_i = \lvert \{ j \leq i \mid \pi(j) \leq \pi(i) \} \rvert.$  
This representation has the advantage that it matches the knowledge of the algorithm. After $i$ arrivals it knows exactly $\rank_1, \ldots, \rank_i$ but it does not know how these candidates compare to future arrivals.
We can therefore apply the principle of deferred decisions and assume that $\rank_{i+1}, \ldots, \rank_n$ will be drawn independently at later points in time.
At time $t$, we call the candidate $\max\{i:A_i \leq t, R_i=1\}$ the \emph{best-so-far}.

Whenever a candidate leaves, the algorithm must irrevocably make a choice to accept or reject based only on the \emph{history} until time $t$. A history is a triple $h=(\arrtimevecvals, \deptimevecvals, \rankvecvals)$ such that, for some $\sofarval$, $\arrtimevecvals \in [0, 1]^\sofarval$ is a vector of arrival times, $\arrtimeval_1 \leq \arrtimeval_2 \leq \ldots \leq \arrtimeval_\sofarval$. $\deptimevecvals \in ([0, t]\, \cup \perp)^\sofarval$ is a vector of departure times (where $\perp$ indicates the candidate has not departed by time $t$), and $\rankvecvals$ are the relative ranks of the $k$ candidates. Note that the history cannot contain the realization of the random variable for the time a candidate stays in the system if he is still there at time $t$. Denote by $\hist_t$ the random variable for the history at time $t$;  $h$ will be used for  realizations of this random variable.

\section{Optimal Stopping Rule}\label{sec:optimal stopping}

In this section, we characterize the optimal stopping rule when the candidates arrive according to distribution $\arrtimedist$ and the waiting time is sampled from  some distribution $\staydist$. 

\begin{definition}
	We call a policy for the  \sdsp \, \emph{bivariate}  if its decision (to accept or reject), given a history $h$, depends only on $t$ and $\numsofar_t$, the number of candidates that have arrived up to time $t$. In other words,  there exists a function $\Theta_n(t, \numbersofar_t)$,
	$$\Theta_{n}\colon [0,1]\times [n] \rightarrow \{0,1\},$$ such that whenever a candidate $x$ departs at time $t$ and $\numbersofar_t$ candidates have arrived prior to $t$, accept if $x$ is the best candidate seen so far and $\Theta_n(t,\numbersofar_t)=1$, otherwise reject. 
\end{definition}

{\renewcommand{\thetheorem}{\ref{thm:opt}}
	\begin{theorem}
		There exists an optimal policy for the \sdsp \, that is bivariate. The function $\Theta_n$ is non-decreasing in $\numbersofar_t$ for fixed $t$ and if $\arrtimedist$ is uniform it is non-increasing in $t$ for fixed $\numbersofar_t$.
	\end{theorem}
	\addtocounter{theorem}{-1}
}

The following two examples show that both $\numbersofar_t$ and $t$ are necessary, i.e., $\Theta$ is indeed a function of both $\numbersofar_t$ and $t$, and not just of one of them. 

\begin{example}[The optimal policy depends on $\numbersofar_t$] \label{ex:needd}
	It is trivial to confirm that for any $t$, any $n>2$, and any distributions, if $\numbersofar_t=1$ the optimal policy rejects, and if $\numbersofar_t=n$ the optimal policy accepts.
\end{example}

\begin{example} [The optimal policy depends on $t$]\label{ex:unif}
	Let the arrival distribution be the uniform distribution, and assume that each candidate stays in the system for some very small fixed time $\eps$.	 If the number of candidates that have arrived by time $4/9$ is $4n/9$, we should accept, as we are virtually in the regular secretary case. To see this most clearly, first notice that in the immediate departure setting, the probability of success when rejecting is at most $\frac{1}{e}$\footnote{The probability of success when rejecting is decreasing in the number of arrivals, while the probability of success when accepting is increasing. Their intersection is roughly at $\frac{n}{e}$ arrivals, where both probabilities are roughly $1/e$.}. Second, the probability of success given the policy rejects at time $t$ is upper bounded by the probability of success given the policy rejects at time $t$ conditioned on candidates overlapping $+$ the probability that candidates overlap. The first term is equal to the probability of success in the instant departure setting, and the second term is upper bounded by $\frac{1}{n}$ by picking $\epsilon \in O(\frac{1}{n^3})$. On the other hand,the probability of success of accepting at time $t$ is $4/9$.

	However, if at time $t = 1-\eps$ the number of candidates that has arrived is $4n/9$, we should reject: all the remaining candidates will arrive by time $t + \eps$, and none of them will depart. Therefore, the probability of success if we reject is equal to the probability that the best candidate is not one of the first $\frac{4n}{9}$ candidates, i.e., $5/9$.
\end{example}

Example~\ref{ex:weirddist} (Appendix~\ref{app:examples}) shows that there exist some  arrival and departure distributions for which the optimal policy's decision is not non-increasing on $t$, i.e., for some $t$ the optimal policy rejects but for some $t' > t$ it accepts. 

Generally, an optimal policy at any point in time makes a decision that maximizes the probability of success from this point onwards. As ties can occur, the policy is not unique. Therefore, we will consider \emph{lazy} policies:  A policy is lazy if  it rejects whenever at a time of departure acceptance and rejection have identical conditional success probabilities. 
Observe that there is always a lazy optimal policy and it is unique. Denote this policy by $\opt$.
We first make some standard observations regarding the optimal lazy policy $\opt$ (see, e.g.,~\cite{Bruss2000,Dynkin63,GM66}). 

\begin{observation}
	$\opt$ only accepts candidates that are best-so-far.
\end{observation}

Given a time $t$, let $\acc_t$ be the policy that accepts only at time $t$ and only in the event that the best-so-far candidate departs at time exactly $t$.
Let  $\rej_t$ be the policy that rejects all departing candidates up to and  including time $t$ and thereafter continues with the optimal policy. 
Given any policy $\pol$, we denote by $\win(\textsc{pol})$ the event that $\pol$ selects the best candidate.

\begin{obs}\label{obs:ob}
	Given a time $t$, let $h$ be a history until $t$ in which a best-so-far candidate departs at $t$. Then $\opt$ accepts this candidate if and only if it has not accepted any candidate before and
	\[
	\Pr\left[\win\left(\rej_{t}\right) \mid \hist_t = h\right] < \Pr\left[\win\left(\acc_t \right) \mid \hist_t = h\right].
	\]
\end{obs}

\begin{obs}
	\label{obs:dn}
	Given a time $t$, let $h$ be a history until $t$ in which $\sofarval$ candidates have arrived by time $t$ and a best-so-far candidate departs at time $t$. Then $\Pr\left[\win\left(\acc_t\right) \mid \hist_t = h\right] = \frac{\sofarval}{n}$.
\end{obs}

Henceforth, when reasoning about $\opt$, we only consider the events when a candidate that is the best out of all those seen thus far  leaves; as $\opt$ never needs to accept at any other time, it suffices to define $\opt$ only on these events.

In order to prove Theorem~\ref{thm:opt}, we need to prove three things: (1) The optimal policy depends only on $t$ and  $\numbersofar_t$ (not on the complete history $\hist_t$) (Lemma~\ref{lem:dt}), (2) the optimal policy is monotone non-decreasing in $\numbersofar_t$ (Lemma~\ref{lem:d}), and (3) the optimal policy is monotone non-decreasing in $t$ for uniform arrivals (Lemma~\ref{lem:t}).

\subsection{The optimal policy is bivariate}

\begin{lemma}\label{lem:dt}
	Given   $n$, the arrival distribution $\arrtimedist$  and  waiting time distribution $\staydist$, the optimal lazy policy can be described by a bivariate function. That is, the decisions depend only on the time $t$ and $\numsofar_t$, the number of candidates that have arrived until time $t$.
\end{lemma}

\begin{proof}
	The proof is by contradiction. If it is not sufficient for $\opt$ to know $\numbersofar_t$ and $t$, then there must be two histories at time $t$ with the same $\numbersofar_t$ for which $\opt$ decides differently. Denote these histories by $h$ and $h'$. As the probability of success if the candidate is accepted is the same in both cases (by Observation~\ref{obs:dn}), it must hold that the probability of success if they reject is different. 
	
	Denote two such histories by $h=(\arrtimevecvals,\deptimevecvals,\rankvecvals)$ and $h'=(\arrtimevecvals',\deptimevecvals',\rankvecvals')$. It is easy to verify that it is possible to transition from $h$ to $h'$ in finitely many steps $h^{(0)}, \ldots, h^{(m)}$ such that $h^{(0)} = h$, $h^{(m)} = h'$ and each pair $h^{(i)}$, $h^{(i+1)}$ differs only in a single entry in the arrivals, departures, or rank vector. As the probability of success when rejecting on $h$ and $h'$ is different, there must be some $i$ such that the probability of success when rejecting on $h^{(i)}=\left( \arrtimevecvals^{(i)},\deptimevecvals^{(i)},\rankvecvals^{(i)}\right) $ and $h^{(i+1)}=\left( \arrtimevecvals^{(i+1)},\deptimevecvals^{(i+1)},\rankvecvals^{(i+1)}\right) $ is different. Consider such an $i$. Assume w.l.o.g. that $\Pr\left[\win\left( \opt\left( h^{(i)}\right) \right)  \growingmid \hist_t = h^{(i)} \right] > \Pr\left[ \win\left( \opt\left( h^{(i+1)}\right) \right)  \growingmid \hist_t = h^{(i+1)} \right] $. 
	
	We first show that $\opt$ rejects on $h^{(i)}$ and accepts on $h^{(i+1)}$.
	Assume that the opposite holds, i.e., $\opt$ accepts at $h^{(i)}$ and rejects at $h^{(i+1)}$. Then we have that 
	
	\begin{align*}\Pr\left[\win\left(\opt\left(h^{(i)}\right)\right) \growingmid \hist_t = h^{(i)}\right]&= \Pr\left[\win\left(\acc_t\right) \growingmid \hist_t = h^{(i)}\right]\\
	& = \Pr\left[\win\left(\acc_t\right) \growingmid \hist_t = h^{(i+1)}\right]\\
	& \leq \Pr\left[\win\left(\rej_t\right) \growingmid \hist_t = h^{(i+1)}\right]\\
	& = \Pr\left[\win\left(\opt\left(h^{(i+1)}\right)\right) \growingmid \hist_t = h^{(i+1)}\right],
	\end{align*}
	
	a contradiction. The first equality follows from the assumption that $\opt$ accepts at $h^{(i)}$. The second is from Observation~\ref{obs:dn}. The third inequality and last equality follows from the assumption that $\opt$ rejects at $h^{(i+1)}$.

	Next, we define a new policy that rejects on history $h^{(i+1)}$ and has better probability of success than $\opt$, a contradiction. 
	If the difference between $h^{(i)}$ and $h^{(i+1)}$ is in $\arrtimeval^{(i)}_j \neq \arrtimeval^{(i+1)}_j$ or $\rankval^{(i)}_j \neq \rankval^{(i+1)}_j$ for some $j$, simply write $h^{(i)}$ instead of  $h^{(i+1)}$ (i.e.,  change the different coordinate), and continue with $\opt\left(h^{(i)}\right)$ (as if the history at this time is $h^{(i)}$). As the marginal probability of every future event is identical for $h^{(i)}$ and $h^{(i+1)}$, the modified algorithm has the same probability of success as it would when the written history is $h^{(i+1)}$. This is in contradiction to the optimality of $\opt\left(h^{(i+1)}\right)$.
	If $\deptimeval^{(i)}_j \neq \deptimeval^{(i+1)}_j$ for some $j$, we consider the three possibilities:
	
	\begin{enumerate}[leftmargin=*]
		\item $\deptimeval^{(i)}_j \neq \perp, \deptimeval^{(i+1)}_j \neq \perp$. The same reasoning as above holds; overwrite $h^{(i+1)}$ by $h^{(i)}$ and continue with $\opt\left(h^{(i)}\right)$.
		\item  $\deptimeval^{(i)}_j \neq \perp, \deptimeval^{(i+1)}_j = \perp$. Similarly, overwrite $h^{(i+1)}$ by $h^{(i)}$; when  candidate $j$ departs (at some time after $t$), ignore her departure. 
		\item If $\deptimeval^{(i)}_j=\perp, \deptimeval^{(i+1)}_j \neq \perp$, overwrite $h^{(i+1)}$ by $h^{(i)}$, i.e., rewrite $\deptimeval^{(i+1)}_j=\perp$, and draw for candidate $j$ a new random variable  $\staytime'_j$ from $\staydist$ conditional on $\staytime'_j \geq t - \arrtime_j$. Informally, now the algorithm ``believes''  that candidate $j$ is still in the system, and at time  $\deptime_j = 
		\min\{\arrtime_j+\staytime'_j, 1\}$ it will simulate candidate $j$ departing.  Because the algorithm never accepts candidate $j$ when it leaves,  this gives the same conditional probability for all possible futures as when the history is $h^{(i)}$, hence $\Pr\left[\win\left(\opt\left(h^{(i)}\right)\right) \mid \hist_t = h^{(i+1)}\right] \geq \Pr\left[\win\left(\opt\left(h^{(i)}\right)\right) \mid \hist_t = h^{(i)}\right]$, in contradiction to the optimality of $\opt\left(h^{(i+1)}\right)$.\qedhere
	\end{enumerate}
\end{proof}

\subsection{Monotonicity in the number of candidates.}\label{sec:fut}

In order to show that the optimal policy is monotone non-decreasing in $\numsofar_t$, we  need to reason about how candidates arrive after time $t$. We summarize these arrivals as follows: The \emph{future after time $t$}, denoted by $\fut_t$, is the vector of all arrivals $\arrtime_{\numsofar_t + 1}, \ldots, \arrtime_n$, durations  $\staytime_{\numsofar_t + 1}, \ldots, \staytime_n$, and indicators $B_{\numsofar_t + 1}, \ldots, B_n$, where $B_i = 1$ if $\rank_i = 1$ and $B_i = 0$ otherwise. That is, $B_i$ is an indicator for the event that the $i$-th arriving candidate is the best one seen so far at the time of arrival. Note that, as $\opt$ is a bivariate policy, $\fut_t$ completely determines whether or not $\opt$ selects the best candidate after time $t$.

We will be interested in the conditional probability spaces given that $\numsofar_t = \sofarval$ for some $\sofarval$. Let us describe two equivalent ways of sampling a conditional future after time $t$ given that $\numsofar_t = \sofarval$. The first way to draw a future at random is the following: Draw $n-k$ times from $\arrtimedist$ conditioned on $(t,1]$ and order them such that $\arrtime_{\sofarval+1} \leq \arrtime_{\sofarval+2} \leq \ldots \leq \arrtime_n$. Furthermore, independently draw $\staytime_i$ from $\staydist$ and set $B_i = 1$ with probability $\frac{1}{i}$. The second way to draw a conditional future at random is the following: draw $n-k$ times from $\arrtimedist$ conditioned on $(t,1]$ and from $\staydist$, and $n-k$ times without repetition from $[n]$. Arbitrarily partition the sampled values into triples: $(\arrtimeval_i,\staytimeval_i,\pi_i)$, where $\arrtimeval_i \in (t,1], \staytimeval_i \geq 0, \pi_i \in [n]$. Sort these triples by $\arrtimeval_i$ and set $B_i = 1$ whenever it $\pi_i$ is smaller than all $\pi_{i'}$ that appear before in the sequence and all values that do not appear at all.

Let $\mathcal{E}_t$ be the event that at time $t$ a candidate leaves and is best so far. We first prove the following lemma.

\begin{lemma}\label{lemma:kkpo}
	For any $t \in [0,1]$ and any $k \in [n-1]$, it holds that
	\begin{equation*}
	\Pr[\win(\rej_t) \mid \numbersofar_t = \sofarval, \mathcal{E}_t] \geq \frac{\sofarval}{\sofarval+1} \Pr[\win(\rej_t) \mid \numbersofar_t = \sofarval+1, \mathcal{E}_t].
	\end{equation*}
\end{lemma}
\begin{proof}

	Let $\textsc{pol}$ be the following policy.
	It rejects all candidates that depart up to (and including) time $t$. At time $t$, it draws an integer $j$ uniformly at random from $\{\numsofar_t+1,\ldots, n\}$ and then numbers the remaining candidates by $\{\numsofar_t+1,\ldots, n\}$, according to their arrival order. $\textsc{pol}$ then executes the optimal policy as if the $j$-th candidate had already arrived and was not the best until time $t$. When a new candidate $x$ arrives, $\textsc{pol}$ discovers whether $x$ is better than all of the candidates it has seen thus far. If he \emph{is} better than all of them, $\textsc{pol}$ assumes he is also better than the $j$-th candidate. As we've already shown, the optimal policy's decision at time $t'$ only depends on $\numsofar_{t'}$ and if the currently departing candidate is the best one so far, so $\textsc{pol}$ can simulate it on the pretended input.
	
	Because $\rej_{t}$ follows the optimal policy after $t$ , we have
	\begin{equation}\label{badoo}
	\Pr[\win(\rej_{t}) \mid \numbersofar_t = \sofarval, \mathcal{E}_t] \geq \Pr[\win(\textsc{pol}) \mid \numbersofar_t = \sofarval, \mathcal{E}_t].
	\end{equation}

	We will  show that 	\begin{equation}
	\Pr[\win(\textsc{pol}) \mid \numbersofar_t = \sofarval, \mathcal{E}_t] \geq \frac{\sofarval}{\sofarval+1} \Pr[\win(\rej_{t}) \mid \numbersofar_t = \sofarval+1, \mathcal{E}_t]. \label{eq22}
	\end{equation}
	
	Combining Inequalities~\eqref{badoo} and~\eqref{eq22} will complete the proof.
	
	Denote by $\I_j$ the event that the $j$-th candidate is not better than the first $k$ candidates.
	In order to prove Inequality~\eqref{eq22},  we show the following: (i) $\Pr[\I_j] =\frac{k}{k+1}$, and (ii) the ``pretend'' futures that $\textsc{pol}$ observes, conditioned on the event that $j$ is not better than the first $k$,  are distributed identically to $\fut_t$, conditioned on $\numsofar_t = k+1$. In other words, let $D_1$ be the distribution of futures that $\pol$ observes if $j$ is not better than the first $k$ candidates; let $D_2$ be the distribution of $\fut_t$, when $\numsofar_t = k+1$. Then $D_1=D_2$.
	
	Statement (i) is true because the probability that a randomly chosen candidate is not the best out of a set of $k+1$ is exactly $\frac{k}{k+1}$. For Statement (ii), we use the second way to draw a conditional future described above. Removing a uniformly selected observation from the ordered sequence is equivalent to removing, for example, the last draw of the unordered tuples $(\arrtimeval_i,\staytimeval_i,\pi_i)$. The event $\I_j$ corresponds to the event that the respective value $\pi_{n - \sofarval}$ is higher than the smallest value that is not drawn. This means, irrespective of where $\pi_{n - \sofarval}$ appears in the ordered sequence, the observations $B_i$ do not change after it is removed. 
	Therefore, as $\textsc{pol}$ and $\rej_t$ make the same decisions on the same observations, and these observations are identically distributed given that the $j$-th candidate is not better than the first  $k$ candidates,
	
	\begin{align*}
	\Pr[\win(\textsc{pol}) \mid \numbersofar_t = \sofarval, \mathcal{E}_t] & \geq \Pr[\win(\textsc{pol}) \wedge \I_j \mid \numbersofar_t = \sofarval, \mathcal{E}_t] \\
	& = \Pr[\I_j] \Pr[\win(\rej_t) \mid \numbersofar_t = \sofarval + 1, \mathcal{E}_t] \\
	& = \frac{\sofarval}{\sofarval+1} \Pr[\win(\rej_{t}) \mid \numbersofar_t = \sofarval + 1, \mathcal{E}_t]. \qedhere
	\end{align*}
\end{proof}

\begin{lemma}\label{lem:d}
	For any number of candidates $n$, time $t$, arrival distribution $\arrtimedist$ and waiting time distribution $\staydist$, the optimal policy's $\Theta_n$ function is monotone non-decreasing in $\numbersofar_t$. 
\end{lemma}	

\begin{proof}
	Assume that there are some $t$ and $\sofarval$ such that $\Theta_n(t, \sofarval)=1$ and $\Theta_n(t,\sofarval+1)=0$, i.e. the optimal policy accepts when $\sofarval$ candidates have arrived and the best so far departs at time $t$, but rejects when $\sofarval+1$ candidates have arrived (and the best so far departs at time $t$).
	Let $\mathcal{E}_t$ be the event that at time $t$ a candidate leaves and is best so far. Then 
	\[
	\Pr[\win(\acc_t) \mid \numbersofar_t = \sofarval, \mathcal{E}_t] = \frac{\sofarval}{n} .
	\]
	Furthermore $\Theta_n(t, \sofarval)=1$ and $\Theta_n(t,\sofarval+1)=0$ mean that
	\begin{align}
	\Pr[\win(\rej_t) \mid \numbersofar_t = \sofarval, \mathcal{E}_t] & < \Pr[\win(\acc_t) \mid \numbersofar_t = \sofarval, \mathcal{E}_t]  = \frac{\sofarval}{n},\label{rrr1}\\
	\Pr[\win(\rej_t) \mid \numbersofar_t = \sofarval+1, \mathcal{E}_t] &\geq \Pr[\win(\acc_t) \mid \numbersofar_t = \sofarval+1, \mathcal{E}_t]
	= \frac{\sofarval+1}{n}.\label{rrr2} \end{align}
	
	\noindent Combining inequalities~\eqref{rrr1} and~\eqref{rrr2} we have that
	$\Pr[\win(\rej_t) \mid \numbersofar_t = \sofarval, \mathcal{E}_t]$ is strictly smaller than $\frac{\sofarval}{\sofarval+1} \Pr[\win(\rej_t) \mid \numbersofar_t = \sofarval+1, \mathcal{E}_t]$, in contradiction to Lemma~\ref{lemma:kkpo}.
\end{proof}

\subsection{Monotonicity in time}\label{sec:app:time}
In the arrival distribution is uniform, we also have monotonicity of $\Theta_n$ in $t$. If we compare two points in time $t < t'$, conditional on $\numsofar_t = \sofarval$ and $\numsofar_{t'} = \sofarval$ respectively, it is easier for the algorithm to succeed from $t'$, that is if there is less time remaining. The reason is that we can pretend all arrivals between $t'$ and $1$ actually appear between $t$ and $1$ by linear scaling. The effect of this linear scaling is that the pretended arrival times appear as if they are uniformly drawn from $(t, 1]$ but the durations for which candidates are larger than if they were drawn from $\staydist$.

Consequently, the optimal policy is more reluctant at later points in time because the probability of success when continuing the sequence is larger. This argument only works if $\arrtimedist$ is indeed uniform. In Example~\ref{ex:weirddist} we show that for a non-uniform distribution the optimal choices might not be monotone  non-increasing in $t$.

\begin{lemma}\label{lem:t}
	If the arrival distribution is uniform, then for any number of candidates $n$, number of candidates that have arrived so far $\numbersofar_t$ and waiting time distribution $\staydist$, the optimal policy's bivariate function $\Theta_n$  is monotone non-increasing in $t$. 
\end{lemma}	

\begin{proof}
	Assume that there are times $t'<t$ such that for some $\sofarval$ we have $\Theta_n(t', \sofarval) = 0$, $\Theta_n(t, \sofarval) = 1$, i.e., the optimal policy accepts at time $t$ if $\sofarval$ candidates have arrived and the best so far departs at time $t$, but rejects for time $t < t'$.
	
	By Observation~\ref{obs:dn}, any algorithm that accepts the best of $\sofarval$ candidates has success probability  $\frac{\sofarval}{n}$. Therefore, it must hold that the success probabilities of rejecting differ.	Letting $\mathcal{E}_t$ denote the event that at time $t$ a candidate leaves and is best so far, we have
	\begin{align}
	\Pr[\win(\rej_{t'}) \mid \numbersofar_{t'} = \sofarval, \mathcal{E}_{t'}] &\geq \Pr[\win(\acc_{t'}) \mid \numbersofar_{t'} = \sofarval, \mathcal{E}_{t'}]\notag \\
	&	= \Pr[\win(\acc_t) \mid \numbersofar_t = \sofarval, \mathcal{E}_t] \notag\\
	& > \Pr[\win(\rej_t) \mid \numbersofar_t = \sofarval, \mathcal{E}_t]. \label{eq:r1}
	\end{align}
	
	Towards a contradiction, we describe for the case $\numbersofar_t = \sofarval$ and $\mathcal{E}_t$ a policy $\textsc{pol}$ that rejects at time $t$ and whose success probability is $\Pr[\win(\rej_{t'}) \mid \numbersofar_{t'} = \sofarval, \mathcal{E}_{t'}]$. This will contradict the optimality of $\rej_t$.
	
	To define a new policy $\textsc{pol}$, we observe that we can safely pretend that candidates leave early. We can set $\tilde{\staytime}_i = \frac{1 - t}{1 - t'} \staytime_i$ and accordingly pretend departure times $\tilde{\deptime}_i = \arrtime_i + \tilde{\staytime}_i$. At the actual departure time $\deptime_i$, we can perform whatever we would have at $\tilde{\deptime}_i$ if the departure times were $\tilde{\deptime}_{\sofarval + 1}, \ldots, \tilde{\deptime}_n$. This is possible because $\frac{1 - t}{1 - t'} \leq 1$.
	
	On input $\arrtime_{\sofarval + 1}, \ldots, \arrtime_n$, $\tilde{\deptime}_{\sofarval + 1}, \ldots, \tilde{\deptime}_n$, $\rank_{\sofarval + 1}, \ldots, \rank_n$, we define a policy as follows. We linearly scale the time interval $(t, 1]$ to $(t', 1]$, defining $\arrtime'_i = \frac{1 - t'}{1 - t} (\arrtime_i - t) + t'$, $\deptime' = \frac{1 - t'}{1 - t} ( \tilde{\deptime}_i - t) + t'$, $\rank_i' = \rank_i$. On these inputs, we run policy $\rej_{t'}$, pretending that $\numbersofar_{t'} = \sofarval$ and $\mathcal{E}_{t'}$.
	
	Observe that as the scaling is linear $\arrtime'_{\sofarval + 1}, \ldots, \arrtime'_n$ is an ordered sequence of $n - \sofarval$ values from the uniform distribution on $(t', 1]$. Furthermore, $\staytime'_i = \frac{1 - t'}{1 - t} \tilde{\staytime}_i = \staytime_i$. Hence, $\staytime'_{\sofarval + 1}, \ldots, \staytime'_n$ are $n - \sofarval$ independent draws from $\staydist$. Consequently, $\arrtime'_{\sofarval + 1}, \ldots, \arrtime'_n, \staytime'_{\sofarval + 1}, \ldots, \staytime'_n, \rank_{\sofarval + 1}, \ldots, \rank_n$ is distributed like a future on after $t'$.
	
	Therefore, the success probability of $\textsc{pol}$ is $\Pr[\win(\rej_{t'}) \mid \numbersofar_{t'} = \sofarval, \mathcal{E}_{t'}]$.
\end{proof}

\section{Optimal Policy in the Limit}\label{SECTION:LIMIT}

We now show that for large $n$ optimal policies have an even simpler structure. To get the optimal performance asymptotically, it is enough to define a time threshold $t^\ast$ such that, irrespective of the number of arrivals, we always accept a departing candidate that is the best-so-far if he departs at time  $t \geq t^\ast$. We call such policies \emph{single-threshold policies}. Note that, if $n$ is small, even if the departure time is instant, policies that set a time threshold can be far from optimal (see Example~\ref{ex3} in Appendix~\ref{app:examples}). Proofs missing from this section can be found in Appendix~\ref{appendix:limit}.

{\renewcommand{\thetheorem}{\ref{theorem:timethresholding}}
	\begin{theorem}
		Given any continuous arrival distribution $\arrtimedist$ and arbitrary departure distribution $\staydist$, there exists a policy $\pol^*$ defined by a threshold $t^\ast$ with the following properties. It accepts a candidate upon leaving at time $t$ if and only if he is the best so far and $t>t^\ast$. For every $\epsilon > 0$, there is an $n_0$ such that for all $n > n_0$, the success probability of $\pol^*$ is at least $\Pr[\win(\opt_n)] - \epsilon$, where $\opt_n$ is the optimal policy on $n$ candidates.
	\end{theorem}
	\addtocounter{theorem}{-1}
}

To define the policy of Theorem~\ref{theorem:timethresholding}, we first make some observations regarding the conditional success probabilities that a policy has to consider when making accept/reject decisions.

Let $\mathcal{E}_t$ denote  the event that the best-so-far candidate leaves at time $t$. 
Let $\pol_t$ be the best policy out of those  that    reject all candidates that arrive until time $t$ (including $t$). Recall that $\rej_t$ is the policy that rejects all  candidates that \emph{depart} until (and including) time $t$, and thereafter continues with the optimal policy.
Note that $\pol_t$ and $\rej_t$ are different, 
but conditioned on $\mathcal{E}_t$ they have the same probability of success, i.e., $\Pr[\win(\pol_t) \mid \mathcal{E}_t] = \Pr[\win(\rej_t) \mid \mathcal{E}_t]$. This is because if $\mathcal{E}_t$ occurs  $\rej_t$ does not accept any candidate that arrived before $t$ either.  Define $P_n(t)$ to be the success probability of an optimal policy for $n$ candidates that does not accept any candidate that arrives until time $t$ (including $t$). In other words, $P_n(t) = \Pr[\win(\pol_t)]$.

\begin{lemma}
	For any fixed $t$ the sequence $(P_n(t))_{n \in \mathbb{N}}$ is non-increasing.
\end{lemma}

\begin{proof}
	For this proof, we use $\pol_t^n$ to denote $\pol_t$ for a fixed $n$. We show that $\pol_t^n$ can also simulate $\pol_t^{n'}$ for $n' \geq n$ while maintaining the same success probability. For this simulation, $\pol_t^n$ pretends that $n' - n$ additional candidates arrive (in the interval $[0,1]$), and that each of these is worse than the real $n$ candidates. It does so by drawing $n' - n$ additional samples from the distributions $\arrtimedist$ and $\staydist$. On the overall $n'$ candidates it runs $\pol_t^{n'}$. Naturally, it selects the best candidate if and only if $\pol_t^{n'}$ does. Therefore $P_n(t) = \Pr[\win(\pol_t^n)] \geq \Pr[\win(\pol_t^{n'})] = P_{n'}(t)$.
\end{proof}

\begin{lemma}
	For all $n \in \mathbb{N}$ and $t < t'$, we have $P_n(t) \geq P_n(t')$. That is, the function $t \mapsto P_n(t)$ 
	is non-increasing in $t$. Furthermore, it is continuous in $t$.
\end{lemma}

\begin{proof}
	Consider an arbitrary $t < t'$. By definition, $\pol_t$ maximizes $\Pr[\win(\hat{\pol})]$ among all policies $\hat{\pol}$ that reject all candidates that arrive until time $t$. The policy $\pol_{t'}$ would also be a policy that rejects all candidates that arrive until time $t$. Therefore $P_n(t) = \Pr[\win(\pol_t)] \geq \Pr[\win(\pol_{t'})] = P_n(t')$ whenever $t \leq t'$.
	
	To show continuity, we observe that 
	\begin{align*}P_n(t) &= \Pr[\win(\pol_t)]\\
	& \leq n(\arrtimedist(t') - \arrtimedist(t)) + \Pr[\win(\pol_{t'})] \\&= n(\arrtimedist(t') - \arrtimedist(t)) + P_n(t'),\end{align*}
	where $\arrtimedist$ is the CDF of the arrival distribution.  This is because $n(\arrtimedist(t') - \arrtimedist(t))$ is an upper bound on the probability of there being some arrival between $t$ and $t'$. Taking the limit $t' \to t$ shows right continuity in $t$ and $t \to t'$ shows left continuity in $t'$.
\end{proof}

Given that $(P_n(t))_{n \in \mathbb{N}}$ is non-increasing and lower-bounded by $0$, the limit $P(t) = \lim_{n \to \infty} P_n(t)$ exists for any fixed $t$. Furthermore, because these functions are continuous, by Dini's theorem, this convergence is uniform, meaning that also  $\lim_{n \to \infty} \sup_{t \in \mathbb{R}} \left\lvert P(t) - P_n(t) \right\rvert = 0$. As a consequence, the function $t \mapsto P(t)$ is non-increasing and continuous.

As $P$ and $\arrtimedist$ are continuous and have the same domain, the range of $\arrtimedist$ is $[0,1]$ and the range of $P$ is $[0, p]$ for some $p \in (0,1)$, they must intersect. That is, there has to be some $t^\ast$ such that $P(t) \geq \arrtimedist(t)$ for $t < t^\ast$ and $P(t) \leq \arrtimedist(t)$ for $t > t^\ast$. This gives rise to the following definition.

\begin{definition}\label{dfn: definition polstar}
	Let $\pol^*$ be a threshold policy such that $\pol^*(t) = 0$ whenever $t \leq t^*$ and $\pol^*(t) = 1$ whenever $t > t^*$ (and the candidate that is the best so far departs).
\end{definition}

We prove that $\pol^*$ is asymptotically optimal. A key component of our proof is to consider policies that are almost optimal in their local choices. We call such policies $\delta$-wrong.

\begin{definition}
	Given $\delta \geq 0$, a \emph{$\delta$-wrong} policy is derived from the optimal policy by changing the choices as follows. Whenever the expected payoff from the two options \emph{accept} and \emph{reject} differs by at most $ \delta$, make an arbitrary choice.
\end{definition}

\begin{restatable}{lemma}{deltawrongbound}
	\label{lem:delta wrong bound}
	Any $\delta$-wrong policy has success probability at least $\Pr[\win(\opt_n)] - 6 \ln(n+1) \delta - \frac{1}{n}$.
\end{restatable}

This lemma is shown by tracing back the errors through the recursive definition of the success probability. We make use of the fact that the number of points at which a best so far leaves is bounded by $O(\log n)$ with high probability.
Furthermore, we make use of the following concentration result. Let $\gamma = 2 \sqrt{n \ln(2n)}$.

\begin{restatable}{lemma}{concentration}
	\label{lem:concentration}
	With probability $1 - \frac{1}{n}$, for all times $t \in \mathbb{R}$, the number of arrivals so far lies in the interval $[\arrtimedist(t) n - \gamma, \arrtimedist(t) n + \gamma]$.
\end{restatable}

The next two lemmas are the technical core of our argument: We compare the conditional probability $\Pr[\win(\rej_t) \mid \mathcal{E}_t, \numsofar_t = \sofarval] $ to the unconditional probability $P_n(t) = \Pr[\win(\pol_t)]$. We show that for suitable choices of $n$, $t$, and $\sofarval$ they are close to each other. Our constructed threshold policy uses the unconditional probability whereas the optimal policy uses the conditional one. Due to this bound, the threshold policy makes reasonably good decisions.

The first step is to remove the conditioning on $\numsofar_t$. The next statement is similar to Lemma~\ref{lemma:kkpo}, and it is also proven  using simulation arguments. The differences are that the error is now additive and we need both upper and lower bounds.

\begin{restatable}{lemma}{lemlargenconditional}
	\label{lem:largen:conditional}
	For any $t$ and any $\sofarval \in [\arrtimedist(t) n - \gamma, \arrtimedist(t) n + \gamma]$, we have
	\[
	\left\lvert \Pr[\win(\rej_t) \mid \mathcal{E}_t, \numsofar_t = \sofarval] - \Pr[\win(\rej_t) \mid \mathcal{E}_t] \right\rvert \leq \frac{2\gamma + 1}{n}  .
	\]
\end{restatable}

The next step is to remove the conditioning on $\mathcal{E}_t$.

\begin{lemma}
	\label{lem:largen:conditional2}
	For any $t$, we have
	\[
	\left\lvert \Pr[\win(\rej_t) \mid \mathcal{E}_t] - P_n(t) \right\rvert \leq (1-\arrtimedist(t))^n  .
	\]
\end{lemma}

\begin{proof}
	Let $\mathcal{X}$ be the event that exactly one candidate leaves at time $t$. (As the arrival distribution is continuous, almost surely no two candidates leave at the same time; therefore, we can ignore these events.) To get the probability of exactly $\sofarval$ candidates arriving by time $t$ conditioned on $\mathcal{X}$, we can consider only $n - 1$ candidates drawing their arrival times. Exactly $\sofarval - 1$ of them have to arrive by time $t$. Therefore,
	\begin{align}
	\Pr[\numsofar_t = \sofarval \mid \mathcal{X}] &=\binom{n-1}{\sofarval-1} (\arrtimedist(t))^{\sofarval-1} (1 - \arrtimedist(t))^{n-\sofarval} =\frac{\sofarval}{n \arrtimedist(t)} \binom{n}{\sofarval} (\arrtimedist(t))^\sofarval (1 - \arrtimedist(t))^{n-\sofarval} \notag \\
	&=\frac{\sofarval}{n \arrtimedist(t)} \Pr[\numsofar_t = \sofarval]. \label{eq:last line}
	\end{align}

	Now, we turn to the probability of $\mathcal{E}_t$ conditioned on $\mathcal{X}$. By the law of total probability, it can be written as
	$
	\Pr[\mathcal{E}_t \mid \mathcal{X}] = \sum_{\sofarval = 0}^n \Pr[\numsofar_t = \sofarval \mid \mathcal{X}] \Pr[\mathcal{E}_t \mid \mathcal{X}, \numsofar_t = \sofarval].
	$
	Note that $\Pr[\mathcal{E}_t \mid \mathcal{X}, \numsofar_t = \sofarval]  = \frac{1}{\sofarval}$ because we condition on exactly $\sofarval$ candidates arriving by time $t$ and one of them leaving at time $t$. Using Equation~\ref{eq:last line}, we get
	\[
	\Pr[\mathcal{E}_t \mid \mathcal{X}] = \sum_{\sofarval = 1}^n \frac{\sofarval}{n \arrtimedist(t)} \Pr[\numsofar_t = \sofarval] \frac{1}{\sofarval} = \frac{1}{n \arrtimedist(t)} \sum_{\sofarval = 1}^n \Pr[\numsofar_t = \sofarval] = \frac{1}{n \arrtimedist(t)} \Pr[\numsofar_t > 0].
	\]
	
	Using Bayes' rule, we can rewrite $\Pr[\numsofar_t = \sofarval \mid \mathcal{E}_t]$ as follows
	\begin{align*}
	\Pr[\numsofar_t = \sofarval \mid \mathcal{E}_t] 
	& = \Pr[\numsofar_t = \sofarval \mid \mathcal{E}_t, \mathcal{X}]  = \frac{\Pr[\numsofar_t = \sofarval \mid \mathcal{X}] \Pr[\mathcal{E}_t \mid \numsofar_t = \sofarval, \mathcal{X}]}{\Pr[\mathcal{E}_t \mid \mathcal{X}]} \\
	& = \frac{\frac{1}{n \arrtimedist(t)} \Pr[\numsofar_t = \sofarval]}{\frac{1}{n \arrtimedist(t)} \Pr[\numsofar_t > 0]}  = \Pr[\numsofar_t = \sofarval \mid \numsofar_t > 0].
	\end{align*}
	
	This allows us to simplify $\Pr[\win(\pol_t) \mid \mathcal{E}_t]$. We use the fact that after conditioning on $\numsofar_t = \sofarval$, the event $\mathcal{E}_t$ is independent of $\win(\pol_t)$: it is only an event concerning the first $\sofarval$ arrivals, all of which happen by $t$, whereas for $\win(\pol_t)$ only the last $n - \sofarval$ arrivals matter, which arrive after $t$. Therefore
	\begin{align*}
	\Pr[\win(\pol_t) \mid \mathcal{E}_t] & = \sum_{\sofarval = 0}^n \Pr[\numsofar_t = \sofarval \mid \mathcal{E}_t] \Pr[\win(\pol_t) \mid \mathcal{E}_t, \numsofar_t = \sofarval] \\
	& = \sum_{\sofarval = 0}^n \Pr[\numsofar_t = \sofarval \mid \numsofar_t > 0] \Pr[\win(\pol_t) \mid \numsofar_t = \sofarval] \\
	& = \Pr[\win(\pol_t) \mid \numsofar_t > 0].
	\end{align*}
	
	\noindent Given this observation, we can relate $P_n(t)$ to the conditional probability $\Pr[\win(\pol_t) \mid \mathcal{E}_t]$:
	\begin{align*}
	P_n(t) & = \Pr[\win(\pol_t)] \\
	& = \Pr[\numsofar_t > 0] \Pr[\win(\pol_t) \mid \numsofar_t > 0] + \Pr[\numsofar_t = 0] \Pr[\win(\pol_t) \mid \numsofar_t = 0] \\
	& = Pr[\numsofar_t > 0] \Pr[\win(\pol_t) \mid \mathcal{E}_t] + \Pr[\numsofar_t = 0] \Pr[\win(\pol_t) \mid \numsofar_t = 0] \\
	& = Pr[\numsofar_t > 0] \Pr[\win(\rej_t) \mid \mathcal{E}_t] + \Pr[\numsofar_t = 0] \Pr[\win(\pol_t) \mid \numsofar_t = 0].
	\end{align*}
	
	\noindent On the one hand, this implies the upper bound of the lemma
	\[
	P_n(t) \leq \Pr[\win(\rej_t) \mid \mathcal{E}_t] + \Pr[\numsofar_t = 0]  .
	\]
	On the other hand, it also implies the lower bound
	\begin{align*}
	P_n(t) &\geq \Pr[\numsofar_t > 0] \Pr[\win(\rej_t) \mid \mathcal{E}_t] &\\
	&\geq \Pr[\numsofar_t > 0] + \Pr[\win(\rej_t) \mid \mathcal{E}_t] - 1  & \text{($x+y-xy \leq 1, \forall x,y \in [0,1]$)}\\
	& = \Pr[\win(\rej_t) \mid \mathcal{E}_t] - \Pr[\numsofar_t = 0].&
	\end{align*}
	
	
	\noindent For the final step, we use $\Pr[\numsofar_t = 0] = (1 - \arrtimedist(t))^n$.
\end{proof}

We are now ready to complete the proof of Theorem~\ref{theorem:timethresholding}.

\begin{proof}[Proof of Theorem~\ref{theorem:timethresholding}]
	Let $\delta = \frac{\epsilon}{18 \ln(n+1)}$. 
	Because of uniform convergence, if $n$ is large enough then $\sup_t \lvert P_n(t) - \sup_{n' \geq n} P_{n'}(t) \rvert < \frac{\delta}{2} = \frac{\epsilon}{36 \ln(n+1)}$. Furthermore, consider $n$ such that $n > \frac{3}{\epsilon}$ and $\frac{3 \gamma + 1 + \sqrt{n}}{n} \leq \frac{\delta}{2} = \frac{\epsilon}{36 \ln(n+1)}$.
	
	We consider the following bivariate policy $\hat{\pol}$. On $(t, \numsofar_t)$ execute $\pol^*(t)$ if $\numsofar_t \in [\arrtimedist(t) n - \gamma, \arrtimedist(t) n + \gamma]$, otherwise execute $\opt(t, \numsofar_t)$. We observe that by Lemma~\ref{lem:concentration} only with probability $\frac{1}{n}$, there is a $t$ for which $\numsofar_t \not\in [\arrtimedist(t) n - \gamma, \arrtimedist(t) n + \gamma]$. Therefore, the executions of $\hat{\pol}$ and $\pol^*$ only differ with probability $\frac{1}{n}$. This means that $\Pr[\win(\pol^*)] \geq \Pr[\win(\hat{\pol})] - \frac{1}{n} \geq \Pr[\win(\hat{\pol})] - \frac{\epsilon}{3}$.
	
	We prove that $\hat{\pol}$ is a $\delta$-wrong policy. Together with Lemma~\ref{lem:delta wrong bound}, we get $\Pr[\win(\hat{\pol})] \geq \Pr[\win(\opt_n)] - 6 \ln(n+1) \delta - \frac{1}{n} \geq \Pr[\win(\opt_n)] - \frac{2 \epsilon}{3}$. So, in combination we have that $\Pr[\win(\pol^*)] \geq \Pr[\win(\opt_n)] - \epsilon$, which implies the theorem.
	
	Consider any $t$ and $\sofarval$ such that $\hat{\pol}(t, \sofarval) \neq \opt(t, \sofarval)$. It suffices to show that for such values of $t$ and $\sofarval$, the probabilities of success for accepting and rejecting the departing candidate (which is a best-so-far candidate by definition) differ only by $\delta$. That is it suffices to show that
	\[
	\lvert \Pr[\win(\acc_t) \mid \mathcal{E}_t, \numsofar_t = \sofarval] - \Pr[\win(\rej_j) \mid \mathcal{E}_t, \numsofar_t = \sofarval] \rvert \leq \delta.
	\]
	
	By construction $\sofarval \in [\arrtimedist(t) n - \gamma, \arrtimedist(t) n + \gamma]$ (otherwise the two policies would be identical). By Observation~\ref{obs:dn} this implies $
	\left\lvert \Pr[\win(\acc_t) \mid \mathcal{E}_t, \numsofar_t = \sofarval] - \arrtimedist(t) \right\rvert = \left\lvert \frac{\sofarval}{n} - \arrtimedist(t) \right\rvert \leq \frac{\gamma}{n},$
	and by Lemma~\ref{lem:largen:conditional} and Lemma~\ref{lem:largen:conditional2}
	\[
	\left\lvert \Pr[\win(\rej_t) \mid \mathcal{E}_t, \numsofar_t = \sofarval] - P_n(t) \right\rvert \leq \frac{2\gamma + 1}{n} + (1 - \arrtimedist(t))^n  .
	\]
	
	In the case that $\opt(t, \sofarval) = 0$ but $\hat{\pol}(t, \sofarval) = \pol^*(t) = 1$, we must have $t > t^*$, i.e. $P(t) < \arrtimedist(t)$ by the definition of $\pol^*(t)$. So also $P_n(t) \leq P(t) + \frac{\delta}{2} < \arrtimedist(t) + \frac{\delta}{2}$. 
	Furthermore, $P(t) \geq \frac{1}{e} - \arrtimedist(t)$; to see this, notice that for every $n$, (1) $P_n(0) \geq \frac{1}{e}$ (since the classic, instant departure setting is a special case), and (2) a feasible policy for $P_n(t)$ is to execute $P_n(0)$ but never accept a candidate that arrives before $t$ (which happens with probability $\arrtimedist(t)$). Therefore, in this case $\arrtimedist(t) > \frac{1}{2e}$, so also $(1 - \arrtimedist(t))^n \leq \frac{1}{\sqrt{n}}$. This gives us
	\begin{align*}
	\Pr[\win(\acc_t) \mid \mathcal{E}_t, \numsofar_t = \sofarval] & \geq \arrtimedist(t) - \frac{\gamma}{n} \geq P_n(t) - \frac{\delta}{2} - \frac{\gamma}{n} \\
	& \geq \Pr[\win(\rej_t) \mid \mathcal{E}_t, \numsofar_t = \sofarval] - \frac{\delta}{2} - \frac{\gamma}{n} - \frac{2 \gamma + 1}{n} - (1 - \arrtimedist(t))^n \\
	& = \Pr[\win(\rej_t) \mid \mathcal{E}_t, \numsofar_t = \sofarval] - \frac{\delta}{2} - \frac{3 \gamma + 1 + \sqrt{n}}{n} \\
	& \geq \Pr[\win(\rej_t) \mid \mathcal{E}_t, \numsofar_t = \sofarval] - \delta  . 
	\end{align*}
	
	In the case that $\opt(t, \sofarval) = 1$ but $\hat{\pol}(t, \sofarval) = \pol^*(t) = 0$ we have $P(t) > \arrtimedist(t)$, meaning that also $P_n(t) > \arrtimedist(t) - \frac{\delta}{2}$. If furthermore $\arrtimedist(t) \geq \frac{1}{\sqrt{n}}$, then $(1 - \arrtimedist(t))^n \leq \frac{1}{\sqrt{n}}$ and so
	\begin{align*}
	\Pr[\win(\acc_t) \mid \mathcal{E}_t, \numsofar_t = \sofarval] & \leq \arrtimedist(t) + \frac{\gamma}{n} \leq P_n(t) + \frac{\delta}{2} + \frac{\gamma}{n} \\
	& \leq \Pr[\win(\rej_t) \mid \mathcal{E}_t, \numsofar_t = \sofarval] + \frac{\delta}{2} + \frac{\gamma}{n} + \frac{2 \gamma + 1}{n} + (1 - \arrtimedist(t))^n \\
	& \leq \Pr[\win(\rej_t) \mid \mathcal{E}_t, \numsofar_t = \sofarval] + \delta  . 
	\end{align*}
	If $\arrtimedist(t) < \frac{1}{\sqrt{n}}$, then we use non-negativity of the probability to get
	\begin{align*}
	\Pr[\win(\acc_t) \mid \mathcal{E}_t, \numsofar_t = \sofarval] & \leq \arrtimedist(t) + \frac{\gamma}{n} \leq \frac{1}{\sqrt{n}} + \frac{\gamma}{n} \leq \delta
	\leq \Pr[\win(\rej_t) \mid \mathcal{E}_t, \numsofar_t = \sofarval] + \delta. 
	\end{align*}

	So, in either case $\left\lvert \Pr[\win(\acc_t) \mid \mathcal{E}_t, \numsofar_t = \sofarval] - \Pr[\win(\rej_t) \mid \mathcal{E}_t, \numsofar_t = \sofarval] \right\rvert \leq \delta$.
\end{proof}

\subsection{Poisson arrivals}\label{subsec:poisson arrivals}
Our result for large numbers of candidates also carries over to the setting in which the arrivals are generated by a Poisson point process. That is, the overall number of candidates $n$ is not fixed but random. One way to interpret the input generation is that first the number of candidates $n$ as well as $\arrtime_1, \ldots, \arrtime_n$ are determined, and only afterwards candidates are assigned to the arrival times in a random permutation.

In a homogeneous Poisson point process with parameter $\delta$, the probability of exactly $\sofarval$ arrivals by time $t$ is given as
$\Pr[\numsofar_t = \sofarval] = \frac{(\delta t)^\sofarval}{\sofarval!} e^{-\delta t}$.
That is, it is given by a Poisson distribution with parameter $\delta t$; both expectation and variance are $\delta t$.

\begin{theorem}\label{thm:poison}
	Consider arrivals generated by a homogeneous Poisson point process with parameter $\delta$ and any departure distribution $\staydist$. There exists a policy $\pol^*$ defined by a threshold $t^\ast$ with the following properties. It accepts a candidate when she is leaving at time $t$ if and only if she is the best so far and $t>t^\ast$. For every $\epsilon > 0$, there is a $\delta_0$ such that for all $\delta > \delta_0$, the success probability of $\pol^*$ is at least $\Pr[\win(\opt_\delta)] - \epsilon$, where $\opt_\delta$ is the optimal policy for parameter $\delta$.
\end{theorem}

The proof of Theorem~\ref{thm:poison} is deferred to Appendix~\ref{appendix:limit}.

\section{Case Study: Poisson Arrivals and Exponential Departures}\label{SEC:EXP}

In this section we consider the following setting. 
Candidates arrive in the time interval  $\left[ 0 , 1 \right]$ according to a Poisson distribution with parameter $\pois$. 
The candidates draw their waiting time from distribution $\mathcal{L}$, an exponential distribution with parameter $\lambda$. An irrevocable decision to hire a candidate must be made by his departure or time $1$, whichever is sooner. 
In Section~\ref{subsec:optima} we give a closed form expression for the optimal probability of selecting the best candidate, when $\pois \rightarrow \infty$. For $\lambda \rightarrow \infty$, the success probability of setting a threshold $\theta$ is $-\theta \ln\theta$ for all $\theta \in [0,1]$, recovering the classical setting~\cite{Ferguson89}. We give the optimal threshold and success probabilities for some natural values of $\lambda$ in Table~\ref{table:15}. Figures~\ref{fig:test1} and~\ref{fig:test2} that show the optimal threshold and the probability of success as a function of $\lambda$.

{\renewcommand{\arraystretch}{1.0}
	\begin{table}[ht]
		\centering
		\begin{tabular}{l l l} 
			\hline 
			Rate $\lambda$      & Optimal threshold $\theta$ & $\Pr[\win]$    \\
			\hline
			$\lambda \rightarrow 0$& any & 1 \\
			$0.5$ & $\approx 0.493$ & $\approx 0.889$ \\
			$1$ & $\approx 0.486$ & $\approx 0.804$ \\
			$2$ & $\approx 0.473$ & $\approx 0.684$ \\
			$10$ & $\approx 0.4411$ & $\approx 0.417$ \\
			$\lambda \rightarrow \infty$ & $e^{-1}$ & $e^{-1}$ \\
			\hline
		\end{tabular} 
		\caption{Optimal threshold and success probability for exponential departures with parameter $\lambda$. }
		\label{table:15}
	\end{table}
}

\begin{figure}[ht]
	\centering
	\begin{minipage}{.5\textwidth}
		\centering
		\begin{tikzpicture}[scale=.7]
		\begin{axis}[xlabel=$\lambda$, ylabel={Optimal threshold}]
		\addplot[mark=none] table
		[x expr=\lineno*0.1-0.1, y=A]
		{optimalthresholds.dat};
		\end{axis}
		\end{tikzpicture}
		\captionof{figure}{Optimal threshold.}
		\label{fig:test1}
	\end{minipage}%
	\begin{minipage}{.5\textwidth}
		\centering
		\begin{tikzpicture}[scale=.7]
		\begin{axis}[xlabel=$\lambda$, ylabel={Optimal probability of success}]
		\addplot[mark=none] table
		[x expr=\lineno*0.1-0.1, y=A]
		{optimalprobabilities.dat};
		\end{axis}
		\end{tikzpicture}
		\caption{Probability of success.}
		\label{fig:test2}
	\end{minipage}
\end{figure}

\subsection{Computing the optimal threshold}\label{subsec:optima}

When $\pois \rightarrow \infty$, we know that an (almost) optimal policy is a single-threshold one, by Theorem~\ref{thm:poison}. 
The next theorem provides a closed form for the success probability for this policy for a given $\theta$.

\begin{theorem}\label{thm:closedform} For the stochastic departing secretary problem with Poisson  arrivals with parameter $\delta \rightarrow \infty$ and  exponential departures with rate $\lambda$,  the success probability for the single-threshold strategy with threshold $\theta$ tends to
	\begin{multline*}
	\frac{1}{\lambda}\left( 1-e^{-\lambda \theta}\right)  + \frac{1}{\theta\lambda^2}\left( e^{\lambda\theta}-1\right) \left(e^{-\lambda \theta} - e^{-\lambda} \right) \\
	+\frac{1}{\lambda}\left(\theta \lambda-1+e^{-\lambda \theta}\right) \left( \ln\left( \frac{1}{\theta} \right) + \frac{1-\theta}{\lambda \theta}  -\int_{x=0}^{\lambda(1-\theta)} \frac{e^{-x}}{(x-\lambda)^2} dx \right).
	\end{multline*}
\end{theorem}

\begin{proof}
	For the proof it will be useful to have slightly different notation. \textbf{We use $A_i$, $L_i$, $D_i$ to refer to the  $i$-th best candidate and not the $i$-th arrival time as in the previous sections.} That is, $A_1$ is the random variable indicating the arrival of the best candidate. $L_1$ denotes the random variable indicating the length of stay and $D_1$ the departure time of the best candidate.
	
	Furthermore, denote the random variable indicating the (true) rank of the best candidate seen in the time interval $[0,t]$ by $\mir_{t}$. We abuse  notation and refer to the candidate by $\mir_{t}$ as well. The meaning will always be clear from context.
	Let $\live(t)$ be the event that $\mir_t$ has not departed by time $t$, and $\dead(t)$ be the event that $\mir_t$ has departed by time $t$.
	
	For large enough $\pois$, we can leverage Theorem~\ref{thm:poison}. In particular, we can henceforth analyze the uniform arrival setting instead of the Poisson arrival setting. We can also infer that the following policy, denoted here by $\pol_\theta$, is asymptotically optimal: if $\mir_{t}$ departs at time $t>\theta$, for some predetermined $\theta$, accept. Otherwise, reject. Our goal is to compute $\theta$ and $\Pr[\win(\pol_\theta)]$.

	We partition the event space into four disjoint events:   
	(1) $A_1\leq\theta$,  (2) $A_1>\theta$ and $\live(\theta)$,  (3) no candidates arrive before $\theta$ (naturally $A_1 > \theta$ in this case), and (4)   $A_1>\theta$ and $\dead(\theta)$. Note that in events (2) and (4) at least one candidate has arrived, and are therefore disjoint from (3). The last event is further broken down into $n-1$ disjoint events, based on the rank of $\mir_\theta$:
	\begin{align}
	\Pr[\win(\pol_\theta)]=&  \Pr[\win(\pol_\theta) \mid A_1 \leq \theta] \cdot \Pr[A_1 \leq \theta] \notag  \\
	+&  \Pr[\win(\pol_\theta) \mid A_1 > \theta \wedge \live(\theta)]\cdot \Pr[A_1 > \theta]\cdot \Pr[\live(\theta)] \notag  \\
	+&  \Pr[\win(\pol_\theta) \mid  A_i > \theta, \forall i \in \{1,\dots,n \}]\cdot \Pr\left[ A_i > \theta, \forall i \in \{1,\dots,n \} \right] \notag  \\
	+ &\sum_{k=1}^{n-1}\Pr[\win(\pol_\theta) \mid A_1 > \theta \wedge \dead(\theta) \wedge \mir_{\theta}=k+1] \cdot \notag \\ &\cdot \Pr[A_1 > \theta]\cdot \Pr[\dead(\theta)]\cdot \Pr[\mir_{\theta}=k+1| A_1 > \theta]\label{blue1}
	\end{align}

	Case (1) is handled by Lemma~\ref{lemma:1} in Appendix~\ref{app:p} and is straightforward. For case (2), if the best candidate that has arrived until $\theta$ has not departed, we select it unless a better candidate arrives before its departure. In order to select the best candidate, we need a contiguous chain between $\theta$ and the arrival of the best candidate. Here we use the memorylessness of the exponential distribution. Case (3) is negligible for any non-zero $\theta$. 
	Case (4) is the most challenging. To analyze this case, we consider the following, suboptimal policy: \emph{Whenever the best-so-far leaves, accept him}. This is identical to setting a threshold  $\theta=0$ in the single-threshold policy; we call this policy  the \emph{no waiting policy}. Setting the threshold at zero makes it easier to precisely analyze the success probability, for any $n$. This step  also relies on the waiting times coming from a memoryless distribution. The analysis of this step is in Appendix~\ref{nowaiting}. We then notice that if we knew the rank $k$ of the best-so-far at $\thresh$, the probability of success thereafter would be exactly the success probability of the no waiting policy with $k$ candidates. To conclude the analysis, we sum over all possible values of $k$.

	Using Claim~\ref{claim:colors} (in Appendix~\ref{app:p}), Equation~\eqref{blue1} simplifies to the following,  where we denote by $\winn(k,[\theta,1])$ the event that we select the best candidate in the no waiting policy, when $k$ candidates arrive, where the candidates arrive over the interval $[\theta,1]$.
	
	\begin{multline*}
	\Pr[\win(\pol_\theta)]=  \frac{1}{\lambda}\left( 1-e^{-\lambda \theta}\right)  + \frac{1}{\lambda^2 \theta} \left( e^{\lambda\theta} - 1 \right) \left( e^{- \lambda \theta} - e^{-\lambda} \right)  + \Pr\left[\winn(n,[\theta,1])\right] \cdot (1-\theta)^n \\
	+  \frac{1}{\lambda} (1-\theta) \left( \lambda\theta - 1  + e^{-\lambda \theta} \right) + \frac{1}{\lambda} \left( \lambda \theta - 1  + e^{-\lambda \theta}  \right)  \sum_{k=2}^{n}\Pr\left[\winn(k,[\theta,1])\right] (1-\theta)^{k}.
	\end{multline*}
	
	The remaining ingredients necessary for further simplification are (i) as $n$ goes to infinity $\Pr\left[\winn(n,[\theta,1])\right] \cdot (1-\theta)^n \rightarrow 0$, (ii) $\lim_{n \rightarrow \infty} \sum_{k=2}^{n} \frac{(k-1)}{\lambda^k } \int_{t=0}^{\lambda(1-\theta)} e^{-t} t^{k-2} dt = \int_{t=0}^{\lambda(1-\theta)} \frac{e^{-t}}{(t-\lambda)^2} dt$, (iii) $\sum_{k=2}^n \frac{(1-\theta)^{k}}{k}$ and  $\sum_{k=2}^n \frac{(1-\theta)^{k-1}}{\lambda}$ converge to $\theta - \ln\theta - 1$ and $\frac{1-\theta}{\lambda \theta}$ respectively. Combining the above ingredients (see Claim~\ref{appendix:subsec:laststeps computation} for the calculations) we have:
	
	\begin{align*}
	\Pr[\win(\pol_\theta)] &= \frac{1}{\lambda}\left( 1-e^{-\lambda \theta}\right)  + \frac{1}{\theta\lambda^2}\left( e^{\lambda\theta}-1\right) \left(e^{-\lambda \theta} - e^{-\lambda} \right) \\
	&+\frac{1}{\lambda}\left(\theta \lambda-1+e^{-\lambda \theta}\right) \left( \ln\left( \frac{1}{\theta} \right) + \frac{1-\theta}{\lambda \theta}  -\int_{x=0}^{\lambda(1-\theta)} \frac{e^{-x}}{(x-\lambda)^2} dx \right). \qedhere
	\end{align*}
\end{proof}

\section{Future directions}


There are several interesting open avenues for related research; we mention two.
The first interesting direction is relaxing the assumption that the system knows when a candidate is about to depart. In this work, we assume that the optimal policy receives a signal when each candidate departs, and is allowed to make a decision thereafter. 
\emph{What happens when we only receive a signal immediately after a candidate's departure?}

The stochastic departure aspect of our model can be applied to virtually all variations of the secretary problem. Of particular interest is the effect of stochastic departures on the matroid secretary problem~\cite{BIK07}. This would possibly have applications to online mechanism design problems, as matroid secretary problems have a strong connection with online auctions e.g.,~\cite{BIKK08}.
\emph{What are the optimal policy and approximation guarantees for  matroid secretary problems with stochastic departures?}

\bibliographystyle{plainnat}
\bibliography{refs}

\newpage
\appendix

\section{Some more examples}\label{app:examples}

\begin{example}[Lemma~\ref{lem:t} does not hold for all arrival distributions]\label{ex:weirddist}
	Consider the following arrival and departure distributions: each candidate arrives in $[0,4\eps]$ or $(6\eps,1]$ w.p.\ $\eps$, and in $(4\eps, 6\eps]$ w.p.\ $1-2\eps$. Candidates stay in the system for $3\eps$. If $4/9$ of the candidates have arrived by $4\eps$, and now the best leaves, we should reject, as $\Pr[\win(\acc_t)] = 4/9$, $\Pr[\win(\rej_t)] \approx 5/9$. However, if $t=4/9$, we should accept,
\end{example}

\begin{example}[Theorem~\ref{theorem:timethresholding} does not hold for small $n$]\label{ex3}
	In this example we show that policies that pick a time threshold are not optimal for small $n$, even for the classic secretary problem with immediate departures (which is  a special case of the problem we study). Consider $3$ candidates with $\arrtime_i \in U[0,1]$. The optimal policy rejects the first candidate, and accept the next best-so-far. The probability of success is $\frac{1}{2}$. To see this most clearly, notice that this policy makes a wrong decision with probability $\frac{1}{3}$, when the best candidate arrives first, and with probability $\frac{1}{6}$, when the order of arrival is $3, 2, 1$.
	
	Given a threshold $\theta$ we compute the probability of success conditioned on an order of arrival:
	\begin{itemize}
		\item $123$: $Pr[\win] = Pr[ FIRST \geq \theta ] = (1-\theta)^3$
		\item $132$: $Pr[\win] = Pr[ FIRST \geq \theta ] = (1-\theta)^3$
		\item $213$: $Pr[\win] = Pr[ SECOND \geq \theta, FIRST \leq \theta ] = \theta (1-\theta)^2$
		\item $231$: $Pr[\win] = Pr[ SECOND \leq \theta, THIRD \geq \theta ] + Pr[ SECOND \geq \theta, FIRST \leq \theta ] = \theta^2 (1-\theta) + \theta (1-\theta)^2 $
		\item $312$: $Pr[\win] = Pr[ FIRST \leq \theta, SECOND \geq \theta ] = \theta (1-\theta)^2$
		\item $321$: $Pr[\win] = Pr[ SECOND \leq \theta, THIRD \geq \theta ] = \theta^2 (1-\theta)$
	\end{itemize}
	
	The probability of success is $\frac{1}{6} \left( 2 (1-\theta)^3 + 3 \theta (1-\theta)^2 + 2 \theta^2 (1-\theta) \right)$.
	This is a decreasing function of $\theta$ in the interval $[0,1]$, therefore the optimum is achieved for $\theta = 0$, which gives a probability of success equal to $\frac{1}{3}$.
\end{example}

\section{Proofs missing from Section~\ref{SECTION:LIMIT}}\label{appendix:limit}

\deltawrongbound*

\begin{proof} 
	To prove the lemma, we explicitly write the Markov Decision Process (MDP) that an optimal policy solves. We have a finite time horizon and an infinite state space. The states have the form $(t, \numsofar_t, b)$, where $t$ denotes the time that has passed, $\numsofar_t$ denotes the number of arrivals so far, and $b \in \{\bot, \textsc{best}, \textsc{other} \}$ denotes whether a candidate has been accepted so far and whether it was the best so far. We can go from a state $(t, \numsofar_t, \textsc{best})$ to a state $(t', \numsofar_{t'}, \textsc{other})$ ($t' > t, \numsofar_{t'} > \numsofar_{t}$), but not the other way around. The states in which we make a decision are the states at which a best-so-far candidate leaves. If we decide to accept, there are no other decision states afterwards. We can equivalently think of the MDP as a tree of depth $n+1$, where the states at depth $d$ are of the form $(t,d,b)$. It will also be useful to think of the actions, ``accept'' and ``reject'', as nodes in this tree that are between depth $d$ and $d+1$ but not belonging to either of them. Therefore, a decision state of depth $d$ has two children, $A_s$ and $R_s$ (for ``accept'' and ``reject'') whose children are the states of depth $d+1$. By the previous observation, $A_s$ does not have predecessor decision states.
	
	A policy is defined at every decision state. The reward of a policy is $1$ if the final state is $(1, n, \textsc{best})$, otherwise it is $0$. We denote by $V(\pol, s, d)$ the expected reward of a policy $\pol$ given that the policy is in state $s$ at depth $d$.
	Let $V_A(\pol, s, d)$ be the expected reward if at (a decision) state $s$ in depth $d$ we take the action \emph{accept} and $V_R(\pol, s, d)$ be the same quantity take the action \emph{reject}. Naturally, $V(\opt, s, d) = \max\{ V_A(\opt, s, d), V_R(\opt, s, d) \}$. Furthermore, $V(\pol,(0,0,\bot),0) = \Pr[\win(\pol)]$.
	
	Now, consider any $\delta$-wrong policy $\pol$. By definition, if $\lvert V_A(\opt, s, d) - V_R(\opt, s, d) \rvert > \delta$, then $\pol$ takes the exact same action as $\opt$, otherwise it is arbitrary. We will transition from $\opt$ to $\pol$ as follows. Let $\pol^{(j)}_J$ be the policy that takes the same action as $\pol$ in the $j$-th, $j+1$-st, $\ldots, J-1$-st decision, otherwise it follows $\opt$. We will show by induction that $V(\pol^{(j)}_J, s, d) \geq V(\opt, s, d) - \delta(J-j)$ for all $j \leq J$ and all states $s$ at depth $d$.
	
	The statement for $j = J$, which is the base case for our induction, is trivial, since $\pol^{(j)}_J$ is exactly $\opt$. To transition from $j+1$ to $j$, we distinguish two cases. If in a state $s$ where a policy has to take its $j$-th action, $\pol$ takes the same action as $\opt$, then the actions of $\pol^{(j)}_J$ and $\pol^{(j+1)}_J$ are identical. Therefore $V(\pol^{(j)}_J, s, d) = V(\pol^{(j+1)}_J, s, d) \geq V(\opt, s, d) - \delta(J-(j+1)) \geq V(\opt, s, d) - \delta(J-j)$.
	Otherwise, if the decisions differ, then $\lvert V_A(\opt, s, d) - V_R(\opt, s, d) \rvert \leq \delta$. 
	Without loss of generality, assume that the optimal policy accepts whereas $\pol$ (and $\pol^{(j)}_J$) rejects. In this case, $V_R(\opt, s, d) \geq V_A(\opt, s, d) - \delta = V(\opt, s, d) - \delta$.
	
	\begin{align}
	V(\pol^{(j)}_J, s, d) &= V_R(\pol^{(j)}_J, s, d) \notag \\
	&= \sum_{\text{state $\hat{s}$ under $R_s$}} V(\pol^{(j)}_J, \hat{s}, d+1) \Pr[\hat{s} | R_s] \label{eq:second} \\
	&= \sum_{\text{state $\hat{s}$ under $R_s$}} V(\pol^{(j+1)}_J, \hat{s}, d+1) \Pr[\hat{s} | R_s] \label{eq:third}\\
	&\geq \sum_{\text{state $\hat{s}$ under $R_s$}} \left( V(\opt, \hat{s}, d+1) - \delta(J-(j+1)) \right) \Pr[\hat{s} | R_s] \label{eq:fourth} \\
	&= \sum_{\text{state $\hat{s}$ under $R_s$}} V(\opt, \hat{s}, d+1) \Pr[\hat{s} | R_s] - \delta(J-(j+1))  \notag \\
	&= V_R(\opt, s, d) - \delta(J-j) + \delta \notag \\
	&\geq V(\opt, s, d) - \delta(J-j). \notag
	\end{align}
	
	To go from line~\ref{eq:second} to line~\ref{eq:third} we use the fact that every action we take below state $s$ (where we take the $j$-th action) is by definition our $j+1$-st or higher action, and therefore $\pol^{(j)}_J$ and $\pol^{(j+1)}_J$ are identical, and have identical probabilities of success. To go from line~\ref{eq:third} to line~\ref{eq:fourth} we use the inductive hypothesis.
	Consequently, we have now shown that
	
	\[
	\Pr[\win(\pol^{(0)}_J)] \geq \Pr[\win(\opt_n)] - J \delta .
	\]
	
	Next, we show that with probability at least $1 - \frac{1}{n}$ we encounter at most $J = 6 \ln(n+1)$ candidates that are best-so-far at the time of their departure. To get this upper bound, we observe that each of these candidates has to have been a best so far at time of arrival. Let $X$ be a random variable denoting the number of best so far at arrival. The probability that the $i$-th arrival is a best-so-far is $\frac{1}{i}$. Therefore, $\mathbb{E}[X] = H_n \leq \ln(n+1)$. As $X$ can be considered a sum of independent 0/1 (indicator) random variables, standard Chernoff bounds\footnote{The version we use is the following. For a random variable $X = \sum_{i=1}^n Y_i$ with expectation $\mu$, where $Y_i \in \{ 0, 1\}$, we have that $\Pr[ X \geq (1+\delta)\mu ] \leq exp(-\delta \mu / 3 )$, for $\delta \geq 1$. } imply

	\[
	\Pr[X \geq 6 \ln(n+1)] \leq \exp\left( - \frac{6 \ln(n+1)}{3} \right) \leq \frac{1}{n^2} .
	\]
	
	That is $\pol^{(0)}_J$ and $\pol$ differ with probability at most $\frac{1}{n}$, implying the lemma.
\end{proof}

\concentration*

\begin{proof} 
	For $t \in \mathbb{R}$, let $Y_t$ be the number of arrivals until time $t$. Observe that $Y_t$ is a sum of $n$ independent 0/1 random variables, each with expectation $F(t)$.
	
	Hoeffding's inequality gives us
	\[
	\Pr\left[\lvert Y_t - n F(t) \rvert \geq \sqrt{n \ln(2n)}\right] \leq 2 \exp\left( - 2 \frac{\left( \sqrt{n \ln(2n)} \right)^2}{n} \right) = \frac{1}{2n^2} \leq n^{-2} .
	\]
	
	Let $t_j = \sup \{ t \in \mathbb{R} \mid F(t) \leq \frac{j}{n} \}$. By union bound, the probability that there is some $j \in \{1, 2, \ldots, n\}$ for which $\lvert Y_{t_j} - n F(t_j) \rvert \geq \sqrt{n \ln(2n)}$ is at most $\frac{1}{n}$.
	
	Consider an arbitrary $t \in \mathbb{R}$ and let $j$ be such that $t_{j-1} < t \leq t_j$. Observe that $Y_t \geq n F(t) + \gamma$ implies that $Y_{t_j} \geq n F(t) + \gamma \geq n (F(t_j) - \frac{1}{n}) + \gamma \geq n F(t_j) -1 + \gamma \geq n F(t_j) + \sqrt{n \ln(2n)}$. Analogously, if $Y_t \leq n F(t) - \gamma$, then also $Y_{t_{j-1}} \leq nF(t_{j-1}) - \gamma \leq nF(t_{j-1}) - \sqrt{n \ln(2n)}$.
	
	In combination,
	\[
	\Pr\left[\exists t \in \mathbb{R}: \lvert Y_t - n F(t) \rvert \geq \gamma\right] \leq \Pr\left[\exists j \in \{1, \ldots, n\}: \lvert Y_{t_j} - n F(t_j) \rvert \geq \sqrt{n \ln(2n)} \right] \leq \frac{1}{n} . \qedhere
	\]
\end{proof}

\lemlargenconditional*

\begin{proof} 
	
	We start by showing that
	\begin{align}
	\Pr[\win(\rej_t) \mid \mathcal{E}_t, \numsofar_t = k+1] \geq \Pr[\win(\rej_t) \mid \mathcal{E}_t, \numsofar_t = k] - \frac{1}{n} . \label{eq: warm up ineq}
	\end{align}
	
	Consider the following (feasible) policy $\textsc{HAL}$ starting at time $t$, when events $\mathcal{E}_t, \numsofar_t = k+1$ have occurred and we have rejected everyone so far (i.e. a policy with success probability is at most the LHS of~\ref{eq: warm up ineq}). Ignore the existence of an arbitrary candidate from the first $k+1$ that have arrived, and draw a new candidate $c$ from the arrival and departure distributions (conditioned on the arrival time being after $t$) with random rank. The probability that $c$ is the best candidate overall is exactly $\frac{1}{n}$. The new policy $\textsc{HAL}$ copies the decisions of  $\win(\rej_t) \mid \mathcal{E}_t, \numsofar_t = k$ with the ``hallucinated'' candidate $c$. If $c$ were a real candidate, then $\textsc{HAL}$ and $\rej_t$ would succeed in exactly the same events. Whenever $\rej_t$ succeeds and the candidate picked is not $c$, then $\textsc{HAL}$ succeeds as well. Therefore, we only need to worry about the outcomes where $\rej_t$ succeeds and the candidate picked is $c$; in all these outcomes $c$ is the best candidate overall, and therefore the total probability is at most $\frac{1}{n}$.

	\begin{align*}
	\Pr[\win(\rej_t) \mid \mathcal{E}_t, \numsofar_t = k ] &= \Pr[\win(\rej_t) \cap \{\text{$c$ is selected} \} \mid \mathcal{E}_t, \numsofar_t = k ] \\
	&\quad + \Pr[\win(\rej_t) \cap \{ \text{ $c$ is not selected } \} \mid \mathcal{E}_t, \numsofar_t = k ]  \\
	&\leq \Pr[ \text{$c$ is the best candidate} ] + \Pr[\win(\textsc{HAL}) \mid \mathcal{E}_t, \numsofar_t = k+1 ]\\
	&\leq \frac{1}{n} + \Pr[\win(\rej_t) \mid \mathcal{E}_t, \numsofar_t = k+1 ],
	\end{align*}
	which is~\ref{eq: warm up ineq} re-arranged. Inequality~\ref{eq: warm up ineq} implies that for all $\sofarval_2 \geq \sofarval_1$ we have
	\begin{align}
	\Pr[\win(\rej_t) \mid \mathcal{E}_t, \numsofar_t = \sofarval_2] \geq \Pr[\win(\rej_t) \mid \mathcal{E}_t, \numsofar_t = \sofarval_1] - \frac{k_2 - k_1}{n} . \label{eq: first ineq}
	\end{align}

	Furthermore, for all $\sofarval_1 \geq \sofarval_2$, we have
	\begin{align}
	\Pr[\win(\rej_t) \mid \mathcal{E}_t, \numsofar_t = \sofarval_2] \geq \Pr[\win(\rej_t) \mid \mathcal{E}_t, \numsofar_t = \sofarval_1] - \frac{\sofarval_1 - \sofarval_2}{n} . \label{eq: second ineq}
	\end{align}
	This is because if $\numsofar_t = \sofarval_2$ it is a feasible policy to drop a random subset of size $\sofarval_1 - \sofarval_2$ of the future arrivals and to simulate the policy for $\numsofar_t = \sofarval_1$. The new policy makes the correct select whenever the one for $\numsofar_t = \sofarval_1$ does, unless one of the dropped candidates is the global best. The latter happens with probability $\frac{\sofarval_1 - \sofarval_2}{n}$.

	Combining inequalities~\ref{eq: first ineq} and~\ref{eq: second ineq} we have for all $\sofarval$ and $\sofarval'$
	\begin{equation}
	\lvert \Pr[\win(\rej_t) \mid \mathcal{E}_t, \numsofar_t = \sofarval] - \Pr[\win(\rej_t) \mid \mathcal{E}_t, \numsofar_t = \sofarval'] \rvert \leq \frac{\lvert \sofarval - \sofarval' \rvert}{n} .
	\label{eq:differentk}
	\end{equation}
	
	We can partition the probability space by all possible outcomes for $\numsofar_t$, to get
	\begin{align*}
	& \Pr[\win(\rej_t) \mid \mathcal{E}_t] \\
	& = \sum_{\sofarval' = 0}^n \Pr[\numsofar_t = \sofarval'] \cdot \Pr[\win(\rej_t) \mid \mathcal{E}_t, \numsofar_t = \sofarval'] \\
	& = \sum_{\sofarval' \not\in [F(t) n - \gamma, F(t) n + \gamma]} \Pr[\numsofar_t = \sofarval'] \cdot \Pr[\win(\rej_t) \mid \mathcal{E}_t, \numsofar_t = \sofarval'] \\ 
	& \qquad + \sum_{\sofarval' \in [F(t) n - \gamma, F(t) n + \gamma]} \Pr[\numsofar_t = \sofarval'] \cdot \Pr[\win(\rej_t) \mid \mathcal{E}_t, \numsofar_t = \sofarval'] .
	\end{align*}
	
	The first term is upper bounded by $$ \sum_{\sofarval' \not\in [F(t) n - \gamma, F(t) n + \gamma]} \Pr[\numsofar_t = \sofarval'] = \Pr\left[\numsofar_t \not\in [F(t) n - \gamma, F(t) n + \gamma]\right],$$ which by Lemma~\ref{lem:concentration} is at most $1/n$. For the second term, we use \eqref{eq:differentk} to get
	\begin{align*}
	&\sum_{\sofarval' \in [F(t) n - \gamma, F(t) n + \gamma]} \Pr[\numsofar_t = \sofarval'] \cdot \Pr[\win(\rej_t) \mid \mathcal{E}_t, \numsofar_t = \sofarval'] &\\
	&\quad \leq \sum_{\sofarval' \in [F(t) n - \gamma, F(t) n + \gamma]} \Pr[\numsofar_t = \sofarval'] \cdot \left( \Pr[\win(\rej_t) \mid \mathcal{E}_t, \numsofar_t = \sofarval] + \frac{\lvert \sofarval-\sofarval' \rvert}{n} \right) &\\
	&\quad \leq \sum_{\sofarval' \in [F(t) n - \gamma, F(t) n + \gamma]} \Pr[\numsofar_t = \sofarval'] \cdot \Pr[\win(\rej_t) \mid \mathcal{E}_t, \numsofar_t = \sofarval] + \frac{2\gamma}{n}
	\end{align*}

	Combining the two bounds we have $$\Pr[\win(\rej_t) \mid \mathcal{E}_t] \leq \Pr[\win(\rej_t) \mid \mathcal{E}_t, \numsofar_t = \sofarval] + \frac{2\gamma + 1}{n}.$$

	For a lower bound, we completely ignore the first term and get
	
	\begin{align*}
	& \Pr[\win(\rej_t) \mid \mathcal{E}_t] \\
	&\geq \sum_{\sofarval' \in [F(t) n - \gamma, F(t) n + \gamma]} \Pr[\numsofar_t = \sofarval'] \cdot \Pr[\win(\rej_t) \mid \mathcal{E}_t, \numsofar_t = \sofarval']\\
	&\geq \sum_{\sofarval' \in [F(t) n - \gamma, F(t) n + \gamma]}  \Pr[\numsofar_t = \sofarval'] \cdot \left( \Pr[\win(\rej_t) \mid \mathcal{E}_t, \numsofar_t = \sofarval] - \frac{\lvert \sofarval-\sofarval' \rvert}{n}\right) \\
	&\geq \Pr\left[\numsofar_t \in [F(t) n - \gamma, F(t) n + \gamma]\right] \cdot \left( \Pr[\win(\rej_t) \mid \mathcal{E}_t, \numsofar_t = \sofarval] - \frac{2\gamma}{n} \right) \\
	&\geq \left(1- \frac{1}{n}\right) \cdot \left( \Pr[\win(\rej_t) \mid \mathcal{E}_t, \numsofar_t = \sofarval] - \frac{2\gamma}{n} \right) \\
	& \geq \Pr[\win(\rej_t) \mid \mathcal{E}_t, \numsofar_t = \sofarval] - \frac{2\gamma}{n} - \frac{1}{n} . \qedhere
	\end{align*}
\end{proof}

\begin{proof}[Proof of Theorem~\ref{thm:poison}]
	We will use the well-known fact (e.g.,~\cite{gallager}) that for any fixed $n$, conditional on $\numsofar_1 = n$, the arrivals are distributed as $n$ independent draws from the uniform distribution on $[0, 1]$.
	Given any $\epsilon > 0$, we use Theorem~\ref{theorem:timethresholding} as follows. The theorem tell us that for $\epsilon' = \frac{\epsilon}{2}$, there is an $n_0$ such that for all $n > n_0$, $\Pr[\win(\pol^*) \geq \Pr[\win(\opt_n)] - \epsilon'$, where $\opt_n$ is the optimal policy on $n$ arrivals from the uniform distribution.
	
	Now choose $\delta_0$ such that $\delta_0 - \sqrt{\frac{\delta_0}{\epsilon'}} > n_0$. This implies that for all $\delta > \delta_0$ by Chebyshev's inequality 
	\[
	\Pr[\numsofar_1 \leq n_0] \leq \Pr\left[\numsofar_1 \leq \delta - \sqrt{\frac{\delta}{\epsilon'}}\right] \leq  \Pr\left[\left\lvert \numsofar_1 - \delta \right\rvert \geq \sqrt{\frac{\delta}{\epsilon'}}\right] \leq \epsilon' .
	\]
	
	Furthermore, 
	\begin{align*}
	\Pr[\win(\opt_\delta)] &= \sum_{n=0}^\infty \Pr[\numsofar_1 = n] \cdot \Pr[\win(\opt_\delta) \mid \numsofar_1 = n] \\
	&\leq \sum_{n=0}^\infty \Pr[\numsofar_1 = n] \cdot \Pr[\win(\opt_n) \mid \numsofar_1 = n],
	\end{align*}
	because in particular $\opt_n$ could simply simulate $\opt_\delta$ for any $n$.
	
	In combination, this gives us
	\begin{align*}
	\Pr[\win(\pol^*)] & = \sum_{n=0}^\infty \Pr[\numsofar_1 = n] \cdot \Pr[\win(\pol^*) \mid \numsofar_1 = n] \\
	& \geq \sum_{n=n_0 + 1}^\infty \Pr[\numsofar_1 = n] \cdot \left( \Pr[\win(\opt_n) \mid \numsofar_1 = n] - \epsilon' \right) \\
	& \geq \sum_{n=0}^\infty \Pr[\numsofar_1 = n] \cdot \left( \Pr[\win(\opt_n) \mid \numsofar_1 = n] - \epsilon' \right) - \Pr[\numsofar_1 \leq n_0] \\
	& \geq \Pr[\win(\opt_\delta)] - 2 \epsilon' = \Pr[\win(\opt_\delta)] - \epsilon . \qedhere
	\end{align*}
\end{proof}

\section{Proofs missing from Section~\ref{SEC:EXP}}
\label{app:p}

We compute the probabilities of the events we used to partition the probability space. These are given by Observations~\ref{obs1} and~\ref{obs2}.

\begin{obs}\label{obs1}
	$\Pr[\mir_{\theta}=k+1 \mid A_1 > \theta] = \theta(1-\theta)^{k-1}.$
\end{obs}
\begin{proof}
	For $k+1$ to be the best in $[0,\theta]$ it must be in $[0,\theta]$, and $1, \ldots, k$ have to be in $(\theta,1]$. By conditioning, this is already true for candidate $1$, so there remain $k-1$, whose arrivals are independently uniformly distributed on $[0, 1]$.
\end{proof}

\begin{obs}\label{obs2}
	$\Pr[\live(\theta)] = \frac{1}{\theta}\int_{x=0}^{\theta} e^{-\lambda(\theta-x)} dx$
\end{obs}

\begin{proof}
	The event that $\live(\theta) = 1$ is the same as the event that $\mir_\theta$ survives from its arrival time $A_{\mir_\theta}$ until $\theta$. Since the best candidate and $\mir_\theta$ have the same departure distribution, it holds that $\Pr[\live(\theta)] = \Pr[\win \mid A_1 \leq \theta] = \frac{1}{\theta}\int_{x=0}^{\theta} e^{-\lambda(\theta-x)} dx$.
\end{proof}

The following lemma takes care of case  (1) in the proof of Theorem~\ref{thm:closedform}.

\begin{lemma}\label{lemma:1}	
	$\Pr[ \win \mid A_1 \leq \theta] = \frac{1}{\theta}\int_{x=0}^{\theta} e^{-\lambda(\theta-x)} dx.$
\end{lemma}
\begin{proof}
	If the best candidate arrives at time $a_1$ before $\theta$, we succeed iff she remains until after $\theta$, i.e., if $L_1 \geq \theta-a_1$: 
	
	\begin{align*}\Pr[\win | A_1 \leq \theta] &= \Pr[ L_1 \geq \theta - a_1| A_1 = a_1, a_1 \leq \theta]\\
	&= \frac{\Pr[ L_1 \geq \theta - a_1 \wedge A_1 = a_1 \wedge a_1 \leq \theta ]}{\Pr[ A_1 \leq \theta]}\\
	& = \frac{ \int_{x=0}^{\theta} f_{ \mathcal{U}[0,1] } (x) e^{-\lambda(\theta-x)} dx } {\theta} \\
	& = \frac{1}{\theta}\int_{x=0}^{\theta} e^{-\lambda(\theta-x)} dx,
	\end{align*}
	where $f_{ \mathcal{U}[0,1] } (x)$ is the pdf of the uniform distribution in $[0,1]$.
\end{proof}

We need the following lemma for case  (2) of the proof of Theorem~\ref{thm:closedform}.
\begin{lemma}\label{lemma:2}
	$\Pr[\win \mid A_1 > \theta, \live(\theta) ] =  \frac{1}{1-\theta}\int_{x=\theta}^{1} e^{-\lambda(x-\theta)} dx.$
\end{lemma}
\begin{proof}
	Let $a_1$ be the arrival time of the best candidate, i.e. $A_1 = a_1$.	
	First, assume that no candidate that is better than  $\mir_\theta$ arrives between $\theta$ and $a_1$. Then we succeed iff the best candidate arrives before $\mir_\theta$ leaves. As the departures distribution is memoryless, given that $\mir_\theta$ hasn't already departed at $\theta$, we can draw the length of its stay $L_{\mir_\theta}$ once again from $\mathcal{L}$ at time $\theta$. 
	
	To remove the assumption above, note that if a candidate $s$, that is better than $\mir_\theta$ arrives between $\theta$ and the departure time of $\mir_\theta$, call it $d$, we can simply use  $d$ as the departure time for $s$ and draw a new departure time for $\mir_\theta$, since distribution $\mathcal{L}$ is memoryless.   
	As the event $\live(\theta)$ is independent of the events $L_1 \geq A_1 - \theta$ and $A_1>\theta$, 
	
	\begin{align*} \Pr[\win|A_1 > \theta \wedge \live(\theta)] &=\Pr[L_1 \geq A_1 - \theta | A_1 > \theta \wedge \live(\theta) ]\\
	&= \frac{\Pr[L_1 \geq A_1-\theta \wedge A_1>\theta]}{\Pr[A_1>\theta]}\\
	&= \frac{1}{1-\theta}\int_{x=\theta}^{1} e^{-\lambda(x-\theta)} dx. \qedhere
	\end{align*}
\end{proof}

Finally, we need the following lemma for case  (4) of the proof of Theorem~\ref{thm:closedform}.

\begin{lemma}\label{lemma:3}	
	$\Pr[\win \mid A_1 > \theta \wedge \dead(\theta) \wedge \mir_{\theta} = 2] = 1$, and for $k \geq 2$:
	$$\Pr[\win \mid A_1 > \theta \wedge \dead(\theta) \wedge \mir_{\theta} = k+1] = \frac{1}{k} + \frac{1}{\lambda(1-\theta)}  - \frac{(k-1)}{\lambda^k (1-\theta)^k} \int_{x=0}^{\lambda(1-\theta)} e^{-x} x^{k-2} dx. $$
\end{lemma}

\begin{proof}
	The first part of the Lemma~\ref{lemma:3} is trivial. In order to prove the second part, we observe the following. Given that $\dead(\theta)$ and $\mir_{\theta} = k+1$, only the best $k$ candidates can be selected by the algorithm, all of which arrive in $[\theta,1]$. Therefore, we can ignore all but the best $k$ candidates and compute the probability of success in a scenario where $k$ candidates arrive uniformly at random in the interval $[\theta,1]$ and we select the first candidate that is the best so far at the time of its departure. Applying Claim~\ref{claim:ll} completes the proof.\end{proof}

The remaining claims of this appendix are technical claims.

\begin{claim}\label{claim:sum of integrals}
	\[\lim_{n \rightarrow \infty} \sum_{k=2}^{n} \frac{(k-1)}{\lambda^k } \int_{t=0}^{\lambda(1-\theta)} e^{-t} t^{k-2} dt = \int_{t=0}^{\lambda(1-\theta)} \frac{e^{-t}}{(t-\lambda)^2} dt \]
\end{claim}

\begin{proof}
	\[ \sum_{k=2}^{n} \frac{(k-1)}{\lambda^k } \int_{t=0}^{\lambda(1-\theta)} e^{-t} t^{k-2} dt = \int_{t=0}^{\lambda(1-\theta)} \frac{e^{-t}}{\lambda^2 } \left( \sum_{k=2}^n \left( \frac{t}{\lambda} \right)^{k-2} (k-1) \right) dt.\]
	
	As $n$ goes to infinity, $\sum_{k=0}^{n-2} \left( \frac{t}{\lambda} \right)^{k} (k+1)$ converges to $ \frac{\lambda^2}{(t-\lambda)^2}$, for $\left|\frac{t}{\lambda}\right|<1$.
	Therefore, the previous expression becomes:
	$$\int_{t=0}^{\lambda(1-\theta)} \frac{e^{-t}}{\lambda^2 } \frac{\lambda^2}{(t-\lambda)^2} dt = \int_{t=0}^{\lambda(1-\theta)} \frac{e^{-t}}{(t-\lambda)^2} dt. $$
\end{proof}

\begin{claim}\label{claim:colors}
	\begin{multline*}
	\Pr[\win]
	=  \frac{1}{\lambda}\left( 1-e^{-\lambda \theta}\right)  + \frac{1}{\lambda^2 \theta} \left( e^{\lambda\theta} - 1 \right) \left( e^{- \lambda \theta} - e^{-\lambda} \right)  + \Pr\left[\winn(n,[\theta,1])\right] \cdot (1-\theta)^n \\
	+  \frac{1}{\lambda} (1-\theta) \left( \lambda\theta - 1  + e^{-\lambda \theta} \right) + \frac{1}{\lambda} \left( \lambda \theta - 1  + e^{-\lambda \theta}  \right)  \sum_{k=2}^{n}\Pr\left[\winn(k,[\theta,1])\right] (1-\theta)^{k}.
	\end{multline*}
	
\end{claim}

\begin{proof}
	
	From Observations~\ref{obs1} and~\ref{obs2} and Lemmas~\ref{lemma:1},~\ref{lemma:2} and~\ref{lemma:3}, we have
	\begin{align*}
	\Pr[\win]=& {\color{purple} \Pr[\win|A_1 \leq \theta] \cdot \Pr[A_1 \leq \theta]} \\
	+& {\color{orange} \Pr[\win|A_1 > \theta \wedge \live(\theta)]\cdot \Pr[A_1 > \theta]\cdot \Pr[\live(\theta)] }\\
	+& {\color{brown} \Pr[\win| A_i > \theta, \forall i \in \{1,\dots,n \}]\cdot \Pr\left[ A_i > \theta, \forall i \in \{1,\dots,n \} \right]} \\
	+& {\color{blue}\Pr[\win|A_1 > \theta \wedge \dead(\theta) \wedge \mir_{\theta}=2] }\\ & {\color{blue}\cdot \Pr[A_1 > \theta]\cdot \Pr[\dead(\theta)]\cdot \Pr[\mir_{\theta}=2| A_1 > \theta]}\\
	+ &\sum_{k=2}^{n}\Pr[\win|A_1 > \theta \wedge \dead(\theta) \wedge \mir_{\theta}=k+1] \cdots \\ &\cdot \Pr[A_1 > \theta]\cdot \Pr[\dead(\theta)]\cdot \Pr[\mir_{\theta}=k+1| A_1 > \theta]\\
	=&  {\color{purple} \int_{x=0}^{\theta} e^{-\lambda(\theta-x)} dx } + {\color{orange} \frac{1}{1-\theta} \int_{x=\theta}^{1} e^{-\lambda(x-\theta)} dx \cdot (1-\theta)\cdot \frac{1}{\theta}\int_{x=0}^{\theta} e^{-\lambda(\theta-x)} dx }\\
	+& {\color{brown} \Pr\left[\winn(n,[\theta,1])\right] \cdot (1-\theta)^n }\\
	+& {\color{blue}1 \cdot (1-\theta) \cdot \left( 1 - \frac{1}{\theta}\int_{x=0}^{\theta} e^{-\lambda(\theta-x)} dx \right) \cdot \theta }\\
	+&\sum_{k=2}^{n}\Pr\left[\winn(k,[\theta,1])\right] (1-\theta) \left( 1-\frac{1}{\theta}\int_{x=0}^{\theta} e^{-\lambda(\theta-x)}dx\right) \theta(1-\theta)^{k-1}\\
	=&  {\color{purple} \frac{1}{\lambda}\left( 1-e^{-\lambda \theta}\right) }  + {\color{orange} \frac{1}{\lambda} \left( 1 - e^{ -\lambda (1-\theta) } \right) \cdot \frac{1}{\lambda \theta} \left( 1-e^{-\lambda \theta}\right) }\\
	+& {\color{brown} \Pr\left[\winn(n,[\theta,1])\right] \cdot (1-\theta)^n }\\
	+& {\color{blue}\theta (1-\theta) \cdot \left( 1 - \frac{1}{\lambda \theta} \left( 1-e^{-\lambda \theta}\right) \right)} \\
	+&\sum_{k=2}^{n}\Pr\left[\winn(k,[\theta,1])\right] \left( 1 - \frac{1}{\lambda \theta} \left( 1-e^{-\lambda \theta}\right) \right) \theta(1-\theta)^{k} \\
	=&  {\color{purple} \frac{1}{\lambda}\left( 1-e^{-\lambda \theta}\right) }  + {\color{orange} \frac{1}{\lambda^2 \theta} \left( e^{\lambda\theta} - 1 \right) \left( e^{- \lambda \theta} - e^{-\lambda} \right) } \\
	+& {\color{brown} \Pr\left[\winn(n,[\theta,1])\right] \cdot (1-\theta)^n }\\
	+& {\color{blue} \frac{1}{\lambda} (1-\theta) \left( \lambda\theta - 1  + e^{-\lambda \theta} \right) }\\
	+& \frac{1}{\lambda} \left( \lambda \theta - 1  + e^{-\lambda \theta}  \right)  \sum_{k=2}^{n}\Pr\left[\winn(k,[\theta,1])\right] (1-\theta)^{k}. \\
	\end{align*}
	
\end{proof}
\begin{claim}\label{appendix:subsec:laststeps computation}
	
	\begin{align*}
	& \frac{1}{\lambda}\left( 1-e^{-\lambda \theta}\right)  + \frac{1}{\lambda^2 \theta} \left( e^{\lambda\theta} - 1 \right) \left( e^{- \lambda \theta} - e^{-\lambda} \right)  + \Pr\left[\winn(n,[\theta,1])\right] \cdot (1-\theta)^n \\
	&+  \frac{1}{\lambda} (1-\theta) \left( \lambda\theta - 1  + e^{-\lambda \theta} \right) + \frac{1}{\lambda} \left( \lambda \theta - 1  + e^{-\lambda \theta}  \right)  \sum_{k=2}^{n}\Pr\left[\winn(k,[\theta,1])\right] (1-\theta)^{k}\\
	=& \frac{1}{\lambda}\left( 1-e^{-\lambda \theta}\right)  + \frac{1}{\theta\lambda^2}\left( e^{\lambda\theta}-1\right) \left(e^{-\lambda \theta} - e^{-\lambda} \right) \\
	&+\frac{1}{\lambda}\left(\theta \lambda-1+e^{-\lambda \theta}\right) \left( \ln\left( \frac{1}{\theta} \right) + \frac{1-\theta}{\lambda \theta}  -\int_{x=0}^{\lambda(1-\theta)} \frac{e^{-x}}{(x-\lambda)^2} dx \right).
	\end{align*}
\end{claim}

\begin{proof}
	
	Taking $n$ to infinity gives $\Pr\left[\winn(n,[\theta,1])\right] \cdot (1-\theta)^n \rightarrow 0$, and we have:
	
	\begin{align*}
	&\Pr[\win] = \frac{1}{\lambda}\left( 1-e^{-\lambda \theta}\right)  + \frac{1}{\theta\lambda^2}\left( e^{\lambda\theta}-1\right) \left(e^{-\lambda \theta} - e^{-\lambda} \right) + \frac{1}{\lambda} (1-\theta) \left( \lambda\theta - 1  + e^{-\lambda \theta} \right) \\
	&\quad+\frac{1}{\lambda}\left(\theta \lambda-1+e^{-\lambda \theta}\right) \lim_{n \rightarrow \infty}\sum_{k=2}^{n} (1-\theta)^{k} \left( \frac{1}{k} + \frac{1}{\lambda(1-\theta)}  - \frac{(k-1)}{\lambda^{k} (1-\theta)^{k}} \int_{x=0}^{\lambda(1-\theta)} e^{-x} x^{k-2} dx  \right) \\
	&\qquad\qquad\quad~=\frac{1}{\lambda}\left( 1-e^{-\lambda \theta}\right)  + \frac{1}{\theta\lambda^2}\left( e^{\lambda\theta}-1\right) \left(e^{-\lambda \theta} - e^{-\lambda} \right) + \frac{1}{\lambda} (1-\theta) \left( \lambda\theta - 1  + e^{-\lambda \theta} \right) \\
	&\quad+\frac{1}{\lambda}\left(\theta \lambda-1+e^{-\lambda \theta}\right) \lim_{n \rightarrow \infty} \sum_{k=2}^{n} \left( \frac{(1-\theta)^{k}}{k} + \frac{(1-\theta)^{k-1}}{\lambda}  - \frac{k-1}{\lambda^{k} } \int_{x=0}^{\lambda(1-\theta)} e^{-x} x^{k-2} dx\right)  \\
	\end{align*}

	Combining Claim~\ref{claim:sum of integrals} with the fact that $\sum_{k=2}^n \frac{(1-\theta)^{k}}{k}$ and  $\sum_{k=2}^n \frac{(1-\theta)^{k-1}}{\lambda}$ converge to $\theta - \ln\theta - 1$ and $\frac{1-\theta}{\lambda \theta}$ respectively, as $n$ goes to infinity, we get that:
	
	\begin{align*}
	\Pr[\win] &= \frac{1}{\lambda}\left( 1-e^{-\lambda \theta}\right)  + \frac{1}{\theta\lambda^2}\left( e^{\lambda\theta}-1\right) \left(e^{-\lambda \theta} - e^{-\lambda} \right) + \frac{1}{\lambda} (1-\theta) \left( \lambda\theta - 1  + e^{-\lambda \theta} \right) \\
	&+\frac{1}{\lambda}\left(\theta \lambda-1+e^{-\lambda \theta}\right) \left( \theta - \ln\theta - 1 + \frac{1-\theta}{\lambda \theta}  -\int_{x=0}^{\lambda(1-\theta)} \frac{e^{-x}}{(x-\lambda)^2} dx \right) \\
	&= \frac{1}{\lambda}\left( 1-e^{-\lambda \theta}\right)  + \frac{1}{\theta\lambda^2}\left( e^{\lambda\theta}-1\right) \left(e^{-\lambda \theta} - e^{-\lambda} \right) \\
	&+\frac{1}{\lambda}\left(\theta \lambda-1+e^{-\lambda \theta}\right) \left( \ln\left( \frac{1}{\theta} \right) + \frac{1-\theta}{\lambda \theta}  -\int_{x=0}^{\lambda(1-\theta)} \frac{e^{-x}}{(x-\lambda)^2} dx \right). \qedhere
	\end{align*}
\end{proof}

\section{The No Waiting Policy}
\label{nowaiting}

Let us consider the modified arrival process in which only $k$ candidates arrive, with arrival times $\arrtime_i \sim_{i.i.d.} U[\theta,1]$. Again, $i$ denotes the rank of the candidate and each candidate $i$ stays for some time $\staytime_i \sim_{i.i.d} \exp(\lambda)$.

Let $\first = \min_{i \in \{ 2, \ldots, k \}} \arrtime_i$ be the random variable indicating the first arrival time of candidates $2, \ldots, k$. Let $f_{\arrtime_1}=\frac{1}{1-\theta}$ be the PDF of $\arrtime_1$. Let $f_{\first}$ be the PDF of $\first$.

\begin{claim}
	$f_{\first}(y) = \frac{k-1}{1-\theta}\left( \frac{1-y}{1-\theta} \right)^{k-2}$
\end{claim}

\begin{proof}
	As $\first$ is the minimum of $k-1$ independent random variables drawn from $U[\theta,1]$, we have
	\[ \Pr[ \first \geq y] = \Pr[\text{candidates $2,\ldots, k$ arrive after time $y$} ] = \left( \frac{1-y}{1-\theta} \right)^{k-1} \]
	So the CDF of $\first$ is $F_\first(y) = \Pr[\first \leq y] = 1- \Pr[ \first \geq y] = 1 - \left( \frac{1-y}{1-\theta} \right)^{k-1}$. Differentiating completes the claim. 
\end{proof}

\begin{claim} \label{claim:ll}The probability of hiring the best candidate with the policy that simply hires the best candidate upon departure with  $k$ secretaries with arrival distribution $U[\theta,1]$ and exponential waiting distribution with parameter $\lambda$ is 
	$$Pr[\win] = 
	\frac{1}{k} + \frac{1}{\lambda(1-\theta)}  - \frac{(k-1)}{\lambda^k (1-\theta)^k} \int_{t=0}^{\lambda(1-\theta)} e^{-t} t^{k-2} dt.$$
\end{claim}

\begin{proof}
	We compute the probability given fixed values of $\arrtime_1$ and $\first$ via
	\begin{align*}
	Pr[\win] &= \int_{y=\theta}^1 \int_{x=\theta}^1 \Pr[ \win \mid \arrtime_1 = x , \first = y ] \cdot f_{\arrtime_1}(x) \cdot f_\first(y) dx dy \\
	&= \frac{k-1}{(1-\theta)^k}\int_{y=\theta}^1 \int_{x=\theta}^1 \Pr[ \win \mid \arrtime_1 = x , \first = y ] \left( 1-y \right)^{k-2} dx dy \\
	&= \frac{k-1}{(1-\theta)^k}\int_{y=\theta}^1 \int_{x=\theta}^y \Pr[ \win \mid \arrtime_1 = x , \first = y ] \left( 1-y \right)^{k-2} dx dy \\
	&\qquad + \frac{k-1}{(1-\theta)^k}\int_{y=\theta}^1 \int_{x=y}^1 \Pr[ \win \mid \arrtime_1 = x , \first = y ] \left( 1-y \right)^{k-2} dx dy \\
	\end{align*}
	
	We observe that if $x \leq y$, the probability of selecting the best candidate is $1$ because it arrives before any other candidate. Therefore
	\begin{align*}
	& \frac{k-1}{(1-\theta)^k}\int_{y=\theta}^1 \int_{x=\theta}^y \Pr[ \win \mid \arrtime_1 = x , \first = y ] \left( 1-y \right)^{k-2} dx dy \\
	& = \frac{k-1}{(1-\theta)^k} \int_{y=\theta}^1 \left( \int_{x=\theta}^y  \left( 1-y \right)^{k-2} dx \right) dy = \frac{k-1}{(1-\theta)^k} \int_{y=\theta}^1 \left( (y-\theta)\cdot \left( 1-y \right)^{k-2} \right) dy = \frac{1}{k}
	\end{align*}		
	
	Now consider that $x > y$. In this case 
	\[
	\Pr[ \win \mid \arrtime_1 = x , \first = y ] =  \Pr[ \first \text{ survives from time $y$ until time $x$} ].
	\]
	
	Since the departure distribution is exponential (i.e. memoryless), if no candidate better than $\first$ arrives in the interval $[y,x]$, then the probability of surviving is $e^{-\lambda(x-y)}$. To remove this assumption notice that, even if a candidate $C$ that is better than $\first$ arrives at some time $\arrtime_C \in [y,x]$, the probability that $C$ remains until $x$ is the same as the probability that $\first$ remains until $x$, conditioned on having stayed until $\arrtime_C$.

	We have that
	$$\int_{x=y}^1 e^{-\lambda(x-y)} (1-y)^{k-2} dx = \frac{1}{\lambda} \left( (1-y)^{k-2} - e^{-\lambda(1-y)}(1-y)^{k-2} \right).$$
	So, we get
	\begin{align*}
	& \frac{k-1}{(1-\theta)^k}\int_{y=\theta}^1 \int_{x=y}^1 \Pr[ \win \mid \arrtime_1 = x , \first = y ] \left( 1-y \right)^{k-2} dx dy \\
	& = \frac{k-1}{(1-\theta)^k}\int_{y=\theta}^1 \int_{x=y}^1 e^{-\lambda(x-y)} \left( 1-y \right)^{k-2} dx dy \\
	& = \frac{k-1}{\lambda (1-\theta)^k} \left( \int_{y=\theta}^1 (1-y)^{k-2} dy - \int_{y=\theta}^1 e^{-\lambda(1-y)}(1-y)^{k-2} dy \right).
	\end{align*}
	
	The first integral is easy to calculate
	$$ \frac{k-1}{\lambda (1-\theta)^k} \int_{y=\theta}^1 (1-y)^{k-2} dy = \frac{(1-\theta)^{k-1}}{\lambda (1-\theta)^k} = \frac{1}{\lambda(1-\theta)}.$$ 
	
	The second integral can be simplified to
	\[ \textstyle
	\int_{y=\theta}^1 e^{-\lambda(1-y)}(1-y)^{k-2} dy = \int_{v=0}^{1-\theta} e^{-\lambda v} v^{k-2} dv = \int_{t=0}^{\lambda(1-\theta)} e^{-t} \left( \frac{t}{\lambda} \right)^{k-2} \frac{dt}{\lambda} = \frac{1}{\lambda^{k-1}} \int_{t=0}^{\lambda(1-\theta)} e^{-t} t^{k-2} dt,
	\]
	completing the proof.\end{proof}

\end{document}